\documentclass[aps,pra,twocolumn,floatfix,amsfonts,10pt]{revtex4-2}
\usepackage{graphicx}
\usepackage{array}
\usepackage{amsthm}
\usepackage{amsmath}
\usepackage{amssymb}
\usepackage{mathrsfs}
\usepackage{txfonts}
\usepackage{pgfplots}
\usepackage{algorithm}
\usepackage{algorithmic}
\usepackage{multirow}
\usepackage{mathtools}
\usepackage{enumerate}
\usepackage{enumitem}
\usepackage{footnote}
\usepackage{xcolor}
\usepackage{bm}
\usepackage[mathscr]{eucal}
\usepackage{pifont}
\usepackage{pgfplots}
\usepackage{braket}
\usepackage{dsfont}
\usepackage{diagbox}
\usepackage{booktabs}
\usepackage{graphics}
\usepackage{appendix}
\newtheorem{definition}{Definition}
\newtheorem{theorem}{Theorem}
\newtheorem{proposition}{Proposition}

\newtheorem{corollary}{Corollary}
\newtheorem{lemma}{Lemma}

\begin{document}
	
\title{How to improve the discrimination power of classically simulable measurements?}
	
\author{Yiran Wang}
\affiliation{School of Mathematics and Statistics, Shaanxi Normal University, Xi'an 710119, China}
	
\author{Yongming Li}
\email{liyongm@snnu.edu.cn}
\affiliation{School of Mathematics and Statistics, Shaanxi Normal University, Xi'an 710119, China}

\begin{abstract}
	Classically simulable measurements (CSMs) constitute an important class of restricted measurements in the odd-prime-dimensional magic resource theory, referring to measurements with postive discrete Wigner functions. 
	Since their discrimination power may be weaker than that of global measurements, it is important to study how to improve the discrimination power of CSMs. 
	In this paper, we consider three approaches to improving the discrimination power of CSMs, including adding magic resources, using quantum catalysts, and using quantum memories. 
	Specifically, we establish an equivalence between measurements with positive discrete Wigner functions and completely positive Wigner-preserving measurement channels, thereby transforming the problem of improving the discrimination power of CSMs into a channel simulation problem using free operations and magic resources. 
	Based on this connection, we derive lower and upper bounds on the exact channel simulation cost and provide a concrete example for which these bounds coincide. 
	Moreover, for a specific ensemble, we analytically determine the optimal discrimination success probability assisted by an arbitrary number of copies of the qutrit Strange state and prove that three copies are necessary and sufficient for perfect discrimination using CSMs. 
	Subsequently, we establish a no-go theorem showing that, when discriminating a pair of states with positive Wigner functions, neither finite-dimensional quantum catalysts nor fixed finite-dimensional quantum memories can improve the optimal asymptotic success rate under CSMs. 
	Finally, by constructing an ensemble consisting of two states with positive Wigner functions, we demonstrate that consumable magic resources can bridge the gap in discrimination power between CSMs and global measurements, whereas reusable quantum resources cannot provide a sustained discrimination advantage.
\end{abstract}
	
\maketitle

\section{Introduction}
In quantum information theory, information is usually encoded in quantum systems \cite{NC2011}. 
When we aim to obtain the relevant information of a quantum state, we need to perform a measurement on it. 
Quantum state discrimination (QSD) is a fundamental task in quantum information theory \cite{BC2009,BK2015}. 
Its central problem is to determine which quantum state has been prepared by performing an appropriate measurement from a given ensemble. 
QSD not only helps us understand the essence of quantum measurements, but also serves as a central tool in many fields, including quantum cryptography \cite{GRT+2002}, quantum communication \cite{GT2007}, quantum sensing \cite{DRC2017}, and quantum machine learning \cite{BWP+2017}.
Meanwhile, QSD mainly includes two discrimination strategies: minimum-error discrimination \cite{H1969,H1973} and unambiguous discrimination \cite{C2001,FDY2004}. 
The former aims to minimize the average error probability and always produces a conclusive outcome, but the result may not be 100\% correct. 
In contrast, the latter allows one to obtain a result with a certain probability and guarantees that this result is never wrong, but needs to introduce an inconclusive measurement outcome.
In this paper, we focus on the minimum-error discrimination strategy. 
Under this strategy, we are interested in the average success probability, i.e., the maximum probability of correctly identifying the state over a given set of measurements. 
When all quantum measurements are allowed, the optimal success probability for binary QSD is given by the celebrated Holevo-Helstrom theorem \cite{H1969,H1973}. 

However, in many scenarios, we can only use a restricted set of measurements.
Therefore, studying the gap in discrimination power between restricted measurements and general global measurements becomes an important problem.
We study an important class of restricted measurements in the magic resource theory in this paper. 
The celebrated Gottesman-Knill theorem stipulates that any stabilizer operation can be efficiently simulated by a classical computer, including applying a Clifford gate, preparing a stabilizer state, discarding subsystems, and measuring in the computational basis \cite{G1999,VMG+2014,CG2019}. 
This theorem indicates that the realization of quantum computational advantage must rely on certain nonclassical resources beyond the stabilizer formalism. 
Such resources are known as nonstabilizerness, or more vividly called magic.
To quantify and manipulate magic resources, the resource theory of magic has been developed in recent years \cite{VFG+2012,HC2017,BCH+2020}. 
In magic resource theory, operations that can be efficiently simulated by classical computers are regarded as free operations, while non-Clifford gates that can inject magic resources or magic states are regarded as resources \cite{BK2005}. 
In odd-prime-dimensional quantum systems, quantum states and quantum operations can be represented by discrete Wigner functions \cite{W1987}. 
Quantum states with positive Wigner functions (PWFs) can be regarded as  classically simulable, whereas states exhibiting Wigner negativity are regarded as magic states \cite{G2007}. 
Correspondingly, positive operator-valued measures (POVMs) composed of effects with PWFs constitute a class of classically simulable measurements (CSMs) \cite{ME2012,ZLZ+2024,WL2026++}. 
Since PWF POVMs form a proper subset of global POVMs, their ability to discriminate quantum states may be weaker than that of unrestricted global measurements. 
This leads to the central question addressed in this paper: how to improve the discrimination power of CSMs?

A natural way to improve the discrimination power is to add magic resources to restricted measurements. 
This idea is closely related to entanglement-assisted state discrimination under locality constraints, where shared entanglement can improve the discrimination power of local operations and classical communication, separable, or positive partial transpose POVMs \cite{C2008,BHN2016,ZZL+2025}. 
Inspired by this entanglement-assisted method, we study magic-assisted discrimination under PWF POVMs. 
More specifically, we consider a binary QSD task, in which the original input state and several fixed magic states are jointly measured by PWF POVMs. 
Apart from directly introducing consumable magic resources into the measurement process, another natural question is whether quantum catalysts and quantum memories can be used to improve the discrimination power of CSMs. 
In quantum resource theories, a quantum catalyst refers to an auxiliary quantum system that participates in a resource transformation process but remains unchanged after the process is completed \cite{JP1999,DKM+2023,LWN2024}. 
Although the catalyst itself is not consumed, it may still make certain state transformations possible that would otherwise be impossible \cite{T2007,KDS2021,LS2021,TS2022,WL2026+}.
Unlike a catalyst, a quantum memory is not required to return to its initial state after each round of transformation \cite{KMR2005}. 
Instead, it can be continuously updated as the protocol proceeds and retain the historical information generated in previous rounds.  
Therefore, it can improve the overall discrimination performance by adaptively adjusting subsequent measurement strategies \cite{CDP2008,HHL+2010}. 
Since catalysts and memories can enhance the success probability of local state discrimination in the resource theory of entanglement \cite{YDY2012,SLN+2022,PS2025}, it is necessary to investigate their roles in the resource theory of magic.

In this paper, we first introduce the framework of the optimal magic-assisted average success probability and the magic cost of optimal discrimination. 
We then introduce measurement channels and establish an equivalence between PWF POVMs and completely positive Wigner-preserving (CPWP) measurement channels. 
This equivalence allows us to transform the problem of implementing measurements beyond CSMs into a channel simulation problem using free operations and magic states. 
Based on this connection, we derive lower and upper bounds on the exact channel simulation cost. 
Specifically, we construct a magic measurement channel and prove that, when the qutrit Strange state is used as the magic resource, its exact simulation cost is equal to one. 
We then show that a single copy of the Strange state strictly improves the optimal success probability of an explicit QSD task under CSMs, and prove that three copies are necessary and sufficient for perfect discrimination of the same ensemble. 
Furthermore, we derive an analytical expression for the optimal success probability of discriminating this ensemble assisted by an arbitrary number of copies of the Strange state, thereby exactly characterizing the relationship between the amount of magic resources and the optimal discrimination success probability.
Second, we formulate the magic-assisted binary state discrimination problem with equal prior probabilities as a semidefinite program (SDP) and derive its dual problem, thereby providing a direct method for computing the optimal success probability for a fixed number of copies of a magic state.
Finally, we investigate whether reusable quantum resources, namely quantum catalysts and quantum memories, can improve the discrimination power of CSMs. 
We introduce a multi-round CPWP discrimination protocol assisted by a finite-dimensional memory system and prove that, for a pair of PWF states, neither finite-dimensional quantum catalysts nor fixed finite-dimensional quantum memories can improve the optimal asymptotic success rate under CSMs. 
To further reveal the different roles of consumable and reusable quantum resources in improving the discrimination power of CSMs, we construct an ensemble consisting of two PWF states. 
For this ensemble, we show that three copies of the Strange state are necessary and sufficient for CSMs to achieve the optimal success probability of global measurements, whereas neither finite-dimensional quantum catalysts nor fixed finite-dimensional quantum memories can provide an asymptotic discrimination advantage. 
This result establishes a strict separation between consumable magic resources and reusable quantum resources in restricted quantum state discrimination.

The remainder of this paper is organized as follows. 
In Section II, we review the stabilizer formalism, the discrete Wigner representation, and several related magic measures. 
In Section III, we study whether consumable magic resources can improve the discrimination power of CSMs. 
We establish the connection between PWF POVMs and CPWP measurement channels, derive lower and upper bounds on the exact simulation cost, and provide a concrete tight example based on the qutrit Strange state. 
Moreover, we determine the exact magic cost of optimal discrimination for a specific ensemble and analytically characterize the optimal success probability assisted by an arbitrary number of copies of the Strange state. 
We also formulate the magic-assisted binary state discrimination problem with equal prior probabilities as an SDP. 
In Section IV, we investigate the roles of quantum memories and quantum catalysts, and prove that, for a pair of PWF states, neither finite-dimensional quantum catalysts nor fixed finite-dimensional quantum memories can improve the optimal asymptotic success rate under CSMs. 
Furthermore, we quantitatively compare the discrimination performance of consumable and reusable quantum resources for the same ensemble consisting of two PWF states. 
Finally, Section V concludes the paper and discusses several open problems.

\section{Preliminaries}

\subsection{The discrete Wigner function}
For an odd-prime number $j$, we consider $n$-qudit Hilbert space $\mathcal{H}=\mathbb{C}^d$ with the dimension $d=j^n$.
Let $\mathbb{Z}_{j} = \{0,1,\cdots,j-1\}$ be the set of congruence classes modulo $j$.
For such systems, $\mathbb{Z}_{j}^{n} \times \mathbb{Z}_{j}^{n}$ plays a role as the corresponding discrete phase space. 
The definitions of the shift operator $X$ and the boost operator $Z$ on the one-qudit Hilbert space $\mathcal{H} = \mathbb{C}^j$ are as follows \cite{S1960}
\begin{equation}
	X \ket{i} = \ket{i+1\mod j}, \quad Z \ket{i} = \omega^i \ket{i}, \quad i \in \mathbb{Z}_{j}, 
\end{equation}
where $\omega = e^{2 \pi i/j}$.
They are called generalized Pauli operators, which generate the discrete Heisenberg-Weyl group (also called the generalized Pauli group \cite{G1997}). 
And they can be expressed in the following forms
\begin{equation}
	X = \sum_{i=0}^{j-1} \ket{i+1} \bra{i}, \quad Z=\sum_{i=0}^{j-1} \omega^i \ket{i} \bra{i}.
\end{equation}
For a point in the discrete phase space $\mathbf{u} = (k,l)$, the corresponding discrete Heisenberg-Weyl operator has the following form 
\begin{equation}
	T_{\mathbf{u}} \coloneqq \tau^{-kl} Z^k X^l,\quad \tau = e^{(j+1)\pi i/j},\quad  k,l \in \mathbb{Z}_{j}.
\end{equation}
The discrete Heisenberg-Weyl operators $T_\mathbf{u}$ can be seen as a set of orthonormal basis of operators on the Hilbert space $\mathcal{H}$, which can generate any operator on the Hilbert space $\mathcal{H}$. 
Now consider a composite system $\mathcal{H}_1 \otimes \cdots \otimes \mathcal{H}_n$ with $\mathrm{dim} \mathcal{H}_1 = d_1, \cdots, \mathrm{dim} \mathcal{H}_{n} = d_n$, the discrete Heisenberg-Weyl operators on the composite system are defined via the tensor products of the discrete Heisenberg-Weyl operators on each subsystem 
\begin{equation}
	T_{\mathbf{u}_1 \oplus \cdots \oplus \mathbf{u}_n} = T_{\mathbf{u}_1}\otimes \cdots \otimes T_{\mathbf{u}_n}.
\end{equation} 
In particular, they generate the discrete Heisenberg-Weyl group
\begin{equation}
	\mathcal{P}_n = \{\tau^a T_{\mathbf{u}}, a \in \mathbb{Z}\}.
\end{equation}

Next, we give the definitions of stabilizer groups and pure stabilizer states on the $n$-qudit system.
We can find a maximal Abelian subgroup $\mathcal{S} \subseteq \mathcal{P}_n$, if it satisfies that $-\mathbf{1}$ is not an element of $\mathcal{S}$, then $\mathcal{S}$ is a stabilizer group and the cardinality of $\mathcal{S}$ is $d$. 
A pure $n$-qudit state $\ket{\psi}$ is a stabilizer state if there exists a stabilizer group $\mathcal{S} \subseteq \mathcal{P}_n$ such that $S \ket{\psi} = \ket{\psi}$ for all $S \in \mathcal{S}$.
The convex hull generated by pure stabilizer states constitutes the stabilizer polytope. 
Moreover, the elements in this convex hull are called STAB.

For each point $\mathbf{u}$, the definition of the phase space point operator $A_{\mathbf{u}}$ is as follows \cite{G2007}
\begin{equation}
	A_\mathbf{0} \coloneqq \frac{1}{d}\displaystyle\sum_{\mathbf{u}}T_{\mathbf{u}}, \quad A_{\mathbf{u}} \coloneqq T_{\mathbf{u}}A_{\mathbf{0}}T_{\mathbf{u}}^{\dagger}.
\end{equation}
Some properties of the phase space point operator $A_{\mathbf{u}}$ are as follows:

\begin{enumerate}
	\item $A_{\mathbf{u}}$ is Hermitian;
	\item $\sum_{\mathbf{u}}A_{\mathbf{u}}/d = \mathds{1}$;
	\item $\mathrm{Tr}(A_{\mathbf{u}}A_{\mathbf{u}'}) = d \delta_{\mathbf{u}\mathbf{u}'}$;
	\item $\mathrm{Tr}(A_{\mathbf{u}}) = 1$;
	\item Given a Hermitian operator $H$, $H = \sum_{u}W_H(\mathbf{u})A_{\mathbf{u}}$;
	\item Given $d$-dimensional operators $P$ and $Q$, $\mathrm{Tr}(PQ)=d\sum_{u}W_P(\mathbf{u})W_Q(\mathbf{u})$.
\end{enumerate}

After defining the phase space point operator, the discrete Wigner function of quantum state $\rho$ at the point $\mathbf{u}$ is given by
\begin{equation}
	W_{\rho}(\mathbf{u}) \coloneqq \frac{1}{d}\mathrm{Tr}(A_{\mathbf{u}}\rho).
\end{equation}
More generally, we can replace quantum state $\rho$ with Hermitian operator $\Delta$ in the definition of the discrete Wigner function.

We say a quantum state $\rho$ has PWFs if it satisfies $W_{\rho}(\mathbf{u})\geqslant0, \forall \mathbf{u} \in \mathbb{Z}_{j} \times \mathbb{Z}_{j}$.
We denote the set of quantum states with PWFs by $\mathcal{W}_+$, i.e.,
\begin{equation}
	\mathcal{W}_+\coloneqq \{\rho:\forall \mathbf{u}, W_\rho(\mathbf{u})\geqslant0, \rho\geqslant 0, \mathrm{Tr}(\rho)=1\}.
\end{equation}
Moreover, it can be obtained that STAB $\subsetneq \mathcal{W}_+$ \cite{VFG+2012,WWS2020}. 
Throughout this work, the free states are the states with positive discrete Wigner functions, denoted by $\mathcal{W}_+$.
A state with Wigner negativity is referred to as a magic state.
For pure states in odd-prime dimensions, the discrete Hudson theorem implies that a state is a PWF state if and only if it is a stabilizer state \cite{G2006}. 

Furthermore, a quantum channel $\mathcal{M}_{A \rightarrow B}$ is called a CPWP channel \cite{WWS2019} if for any system $R$ with odd-prime dimensions, the following holds
\begin{equation}
	\forall \rho_{RA}\in \mathcal{W}_+, \quad (\mathds{1}_R \otimes \mathcal{M}_{A \rightarrow B})\rho_{RA} \in \mathcal{W}_+.
\end{equation}
In this paper, we will denote the set of CPWP channels by $\mathcal{A}$.
A superchannel is a linear map that maps a quantum channel to
another \cite{CDP2008+}.
We call $\Theta$ a completely CPWP-preserving superchannel \cite{WWS2019} if it satisfies
\begin{equation}
	\begin{aligned}
		\forall \mathcal{M}_{RA\rightarrow RB} \in \mathcal{A}&(RA\rightarrow RB), \\ &\Theta_{(A\rightarrow B)\rightarrow(C\rightarrow D)}\left(\mathcal{M}_{RA\rightarrow RB}\right) \in \mathcal{A}(RC\rightarrow RD),
	\end{aligned}
\end{equation}
where $R$ is an arbitrary reference system.

In addition, for an effect $E$ of a POVM, the corresponding discrete Wigner function is defined as follows
\begin{equation}
	W(E|\mathbf{u}) \coloneqq \mathrm{Tr}(A_{\mathbf{u}}E).
\end{equation}
Similarly, the POVM with $n$ effects $\mathbf{E} = \{E_i\}_{i=1}^{n}$ is a PWF POVM if each effect $\{E_i\}$ has PWFs \cite{PWB2015,ZLZ+2024}.
For each effect $\{E_i\}$ in POVM $\mathbf{E} = \{E_i\}_{i=1}^{n}$, it has the following quasi-probability distribution of discrete Wigner function over the phase space 
\begin{equation}
	\displaystyle\sum_{i}W(E_i|\mathbf{u}) = 1.
\end{equation}
Refs. \cite{VFG+2012,G2006} pointed out that only negative quasi-probability distribution can accelerate the calculation of quantum computation.
Therefore, we regard PWF POVMs as CSMs \cite{HWV+2014}, because they can be effectively simulated on a classical computer, and they cannot realize the computational acceleration in quantum computation.

\subsection{Magic measures}
The mana $M$ of a quantum state $\rho$ was first introduced by Ref. \cite{VMG+2014}, and its specific definition is as follows
\begin{equation}\label{definition of state mana}
	M(\rho) \coloneqq \log_2 \displaystyle\sum_{\mathbf{u}} |W_{\rho}(\mathbf{u})|,
\end{equation}
where $W_{\rho}(\mathbf{u})$ is the discrete Wigner function of $\rho$.
Given two quantum states $\rho$ and $\sigma$, it is easy to verify that the mana of quantum states satisfies the following properties:

(i)~\emph{Faithfulness}:~$M(\rho) \geqslant 0$ and $M(\rho) = 0$ if and only if $\rho$ is a PWF state.

(ii)~\emph{Monotonicity:}~$M(\rho) \geqslant M[\mathcal{M}(\rho)]$ for any CPWP channel $\mathcal{M}$.

(iii)~\emph{Additivity:}~$M(\rho \otimes \sigma) = M(\rho) + M(\sigma)$.

Moreover, the definition of the mana $M$ of a quantum channel $\mathcal{M}$ was first introduced by Ref. \cite{WWS2019}, which was defined as follows
\begin{equation}
	M(\mathcal{M}) \coloneqq \log_2 \max_{\mathbf{u}} \displaystyle\sum_{\mathbf{v}} |W_{\mathcal{M}}(\mathbf{v}|\mathbf{u})|,
\end{equation}
where $W_{\mathcal{M}}(\mathbf{v}|\mathbf{u})$ is the discrete Wigner function of $\mathcal{M}$.

Given two quantum channels $\mathcal{M}$ and $\mathcal{N}$, it is easy to verify that the mana of quantum channels satisfies the following properties:

(i)~\emph{Faithfulness}:~$M(\mathcal{M}) \geqslant 0$ and $M(\mathcal{M}) = 0$ if and only if $\mathcal{M}$ is a CPWP channel.

(ii)~\emph{Monotonicity:}~$M(\mathcal{M}) \geqslant M[\Theta(\mathcal{M})]$ for any CPWP-preserving superchannel $\Theta$.

(iii)~\emph{Additivity:}~$M(\mathcal{M} \otimes \mathcal{N}) = M(\mathcal{M}) + M(\mathcal{N})$.

\section{Using magic resources to improve the discrimination power}
Given a quantum state ensemble $\left\{\left(p_i, \rho_i\right)\right\}_{i=0}^{n-1}$, where $p_i \geqslant 0$, $\sum_{i}p_i = 1$, and $\rho_i$ is a quantum state.
QSD refers to the task of giving an ensemble and randomly extracting a state from it, and then determining which state is extracted by a POVM $\mathbf{M}=\left\{M_i\right\}_{i=0}^{n-1}$.
We associate the $i$-th measurement outcome with the $i$-th quantum state. 
In the strategy of minimum-error QSD, our goal is to maximize the average success probability, which is given by
\begin{equation}
	p_{\rm{succ}}^{\mathbf{M}}(\left\{p_i, \rho_i\right\}_{i=0}^{n-1}) = \max_{\{M_i\}\in\mathbf{M}}\displaystyle\sum_{i=0}^{n-1}p_i\mathrm{Tr}(M_i\rho_i).
\end{equation}
In particular, we consider the task of binary QSD in this paper.
Using the celebrated Holevo-Helstrom theorem \cite{H1969,H1973}, the optimal success probability of discriminating two quantum states $\{\rho_0,\rho_1\}$ with prior probability $p_0,p_1$ is given by
\begin{equation}\label{Unrestricted QSD}
	\begin{aligned}
		p_{\rm{succ}}^{\mathbf{M}}(\{p_i,\rho_i\}_{i=0}^1) &= \max_{\{M_i\}\in\mathbf{M}}\displaystyle\sum_{i=0}^{1}p_i\mathrm{Tr}(M_i\rho_i) \\ &= \frac{1}{2}\left( \left\|p_0\rho_0-p_1\rho_1 \right\|_1+1 \right),
	\end{aligned}
\end{equation}
where $\left\|\cdot\right\|_1$ denotes the trace norm.

Next, let us review the definition of distinguishability norm \cite{MWW2009}.
\begin{definition}
	Let $\mathcal{H}$ be a finite-dimensional Hilbert space, $\mathbf{E}=\{E_i\}_{i=0}^{n-1}$ be a set of allowed measurement operators, typically corresponding to a set of restricted measurements.  
	Then the distinguishability norm of any Hermitian operator $\Delta$ with respect to $\mathbf{E}$ is defined as follows
	\begin{equation}
		\|\Delta\|_{\mathbf{E}} \coloneqq \max_{\{E_i\}\in\mathbf{E}} \sum_i \left| \mathrm{Tr}(E_i\Delta) \right|.
	\end{equation}
\end{definition}

After defining the distinguishability norm, we can get the optimal success probability of binary QSD task using a set of restricted measurements $\mathbf{E}$ as follows
\begin{equation}
	p_{\rm{succ}}^{\mathbf{E}}(\{p_i,\rho_i\}_{i=0}^1) = \frac{1}{2}( \left\|p_0\rho_0-p_1\rho_1 \right\|_{\mathbf{E}}+1 ).
\end{equation}
Specifically, when $\mathbf{E}$ is the full set of POVMs acting globally on $\mathcal{H}$, $\left\|p_0\rho_0-p_1\rho_1 \right\|_{\mathbf{E}} = \left\|p_0\rho_0-p_1\rho_1 \right\|_1$.

In magic resource theory, restricted measurements mainly include stabilizer measurements and PWF measurements.
We call stabilizer POVMs STAB POVMs since they only use stabilizer resources \cite{MLF+2025}.
Because of the mathematical complexity of STAB POVMs, we mainly consider the case that the set of restricted measurements are PWF POVMs in this paper.
Since STAB POVMs $\subseteq$ PWF POVMs $\subseteq$ GLOBAL POVMs \cite{ZLZ+2024}, we have
\begin{equation}
	\|\Delta\|_{\mathrm{STAB}} \leqslant \|\Delta\|_{\mathbf{PWF}} \leqslant \|\Delta\|_{\mathrm{GLOBAL}} = \|\Delta\|_1
\end{equation}
for any Hermitian operator $\Delta$.
As the constraints on the measurements vanish, we find that the optimal average success probability of discriminating two quantum states increases, which implies that $p_{\rm{succ}}^{\mathbf{STAB}}(\{p_i,\rho_i\}_{i=0}^1) \leqslant p_{\rm{succ}}^{\mathbf{PWF}}(\{p_i,\rho_i\}_{i=0}^1) \leqslant p_{\rm{succ}}^{\mathbf{GLOBAL}}(\{p_i,\rho_i\}_{i=0}^1)$.

Next, a natural question is: if we are only allowed to use PWF POVMs, how can we make them have the same discrimination power as general global measurements, i.e. $\|\Delta\|_{\mathbf{PWF}} = \|\Delta\|_{\mathrm{GLOBAL}}$?
First, we consider adding magic resources to the system to achieve global measurements. 
Inspired by Refs. \cite{ZZL+2025,WL2026}, we give the optimal magic-assisted average success probability where two parties are allowed to share magic resources to assist their QSD tasks.
\begin{definition}
	The optimal magic-assisted average success probability of discriminating two quantum states $\{\rho_0,\rho_1\}$ with prior probability $p_0,p_1$ by using restricted measurements is given by
	\begin{equation}\label{Magic-assisted QSD}
		\begin{aligned}
			p_{\rm{succ}}^{\mathbf{E}}&\left(\{p_i,\rho_i\}_{i=0}^1, \omega^{\otimes n} \right) = \max_{\{E_i\}\in\mathbf{E}}\displaystyle\sum_{i=0}^{1}p_i\mathrm{Tr}\left[E_i\left(\rho_i\otimes \omega^{\otimes n}\right)\right] \\ &= \frac{1}{2}\left[\left\|p_0\left(\rho_0 \otimes \omega^{\otimes n}\right)-p_1\left(\rho_1 \otimes \omega^{\otimes n} \right) \right\|_{\mathbf{E}}+1 \right],
		\end{aligned}
	\end{equation}
	where $\mathbf{E} \in \left\{\mathrm{STAB~POVMs}, \mathrm{PWF~POVMs} \right\}$ is a set of restricted measurements, $\omega$ is a magic state.
\end{definition}

Then, the next question is: how many magic resources do we need to add to the system so that the optimal magic-assisted average success probability equals the original optimal average success probability of using GLOBAL POVMs, i.e., Eq. \eqref{Magic-assisted QSD} equals Eq. \eqref{Unrestricted QSD}?
We give the definition of magic cost of optimal $\mathbf{E}$-discrimination based on this question.
\begin{definition}
	Given a set of restricted measurements $\mathbf{E}$ and two quantum states $\{\rho_0,\rho_1\}$ with prior probability $p_0,p_1$, the magic cost of optimal $\mathbf{E}$-discrimination is defined as
	\begin{equation}
		\begin{aligned}
			E_c^{\mathbf{E}} = \min \biggl\{ &n:\left\|p_0\rho_0 - p_1\rho_1 \right\|_1  \\ &~~~~~~= \|p_0(\rho_0 \otimes \omega^{\otimes n})-p_1(\rho_1 \otimes \omega^{\otimes n}) \|_{\mathbf{E}} \biggr\},
		\end{aligned}
	\end{equation}
	where $\omega$ is a magic state.
\end{definition}

Since quantum measurements are a special class of quantum channels \cite{NC2011}, measurements can be represented as quantum-classical channels (measurement channels) \cite{H1998,P2009}, which form a special class of quantum channels. 
A measurement channel is defined as follows
\begin{definition}
	Given a POVM $\mathbf{M}=\left\{M_i\right\}_{i=0}^{n-1}$ which satisfies $M_i \geqslant 0$ and $\sum_{i}M_i = \mathds{1}$, the measurement channel $\mathcal{E}$ associated to $\mathbf{M}$ is defined as
	\begin{equation}
		\mathcal{E}(\rho) = \displaystyle\sum_{i} \mathrm{Tr}(M_i \rho) \ket{i} \bra{i},
	\end{equation}
	where $\{\ket{i}\}$ is an orthonormal basis.
\end{definition} 

Next, we show that the measurement channel associated to PWF POVM is a CPWP channel.
To avoid ambiguity, we will use subscripts to indicate the systems whenever necessary in this paper.

\begin{theorem}\label{theorem 1}
	Let $\mathbf{E}=\{E_k\}_{k=0}^{d_B-1}$ be a POVM on an odd-prime-dimensional system $A$ with odd-prime number $d_B$. 
	Define the measurement channel
	\begin{equation}
		\mathcal E_{A\rightarrow B}(\rho) = \sum_{k=0}^{d_B-1} \mathrm{Tr}(E_k\rho)|k\rangle\langle k|_B,
	\end{equation}
	where $\{|k\rangle_B\}_{k=0}^{d_B-1}$ is chosen as a stabilizer basis.
	Then $\mathcal{E}_{A \rightarrow B}$ is a CPWP channel if and only if $\mathbf{E}=\{E_k\}_{k=0}^{d_B-1}$ is a PWF POVM.
\end{theorem}
\begin{proof}
	Given a quantum channel $\mathcal{M}_{A \rightarrow B}$, its corresponding Choi–Jamiołkowski matrix is defined as follows \cite{C1975,J1972}
	\begin{equation}\label{Choi matrix}
		J_{AB}^{\mathcal{M}} = \displaystyle\sum_{i,j}\ket{i}\bra{j}_A\otimes \mathcal{M}(\ket{i}\bra{j}_{A'}),
	\end{equation}
	where $A'$ is a replica of system $A$, $\{\ket{i}_A\}$ and $\{\ket{i}_{A'}\}$ are orthonormal bases on Hilbert spaces $\mathcal{H}_A$ and $\mathcal{H}_{A'}$.
	According to Eq. \eqref{Choi matrix}, we can get the Choi–Jamiołkowski matrix of measurement channel $\mathcal{E}_{A \rightarrow B}(\rho) = \sum_{k} \mathrm{Tr}(E_k \rho) \ket{k} \bra{k}_B$ associated to $\mathbf{E}=\{E_k\}_{k=0}^{d_B-1}$ as follows
	\begin{equation}
		\begin{aligned}
				J_{AB}^{\mathcal{E}} &= \displaystyle\sum_{i,j,k}\ket{i}\bra{j}_A\otimes \mathrm{Tr}(E_k\ket{i}\bra{j})\ket{k}\bra{k}_B \\ &= \displaystyle\sum_{i,j,k}\langle j|E_k|i \rangle \ket{i}\bra{j}_A\otimes \ket{k}\bra{k}_B \\ &= \displaystyle\sum_{k} E_k^{\top}\otimes \ket{k}\bra{k}_B,
		\end{aligned}
	\end{equation}
	where $\{\ket{k}_B\}$ is a stabilizer basis on Hilbert space $\mathcal{H}_B$.
	
	Meanwhile, Ref. \cite{WWS2019} proved that a quantum channel is CPWP if and only if the discrete Wigner function of its corresponding Choi–Jamiołkowski matrix is positive.
	The discrete Wigner function of $\mathcal{E}$ is defined as follows
	\begin{equation}
		\begin{aligned}
			W_{\mathcal{E}}(\mathbf{v}|\mathbf{u}) &= \frac{1}{d_B} \mathrm{Tr}\left\{ \left[\left(A_A^{\mathbf{u}}\right)^{\top} \otimes A_B^{\mathbf{v}}\right]J_{AB}^{\mathcal{E}}\right\} \\ &= \frac{1}{d_B} \displaystyle\sum_{k}\mathrm{Tr}\left(A_A^{\mathbf{u}}E_k\right) \langle k|A_B^{\mathbf{v}}|k \rangle \\ &= \displaystyle\sum_{k} W(E_k|\mathbf{u})W_{\ket{k}\bra{k}}(\mathbf{v}),
		\end{aligned}
	\end{equation}
	where the second equality follows from the fact that $\mathrm{Tr}(A^{\top}B^{\top}) = \mathrm{Tr}(BA) = \mathrm{Tr}(AB)$, the last equality follows from the facts that $W(E_k|\mathbf{u}) = \mathrm{Tr}\left(A_A^{\mathbf{u}}E_k\right)$ and $W_{\ket{k}\bra{k}}(\mathbf{v}) = 1/(d_B)\mathrm{Tr}\left(A_B^{\mathbf{v}}\ket{k}\bra{k}\right)$.
	Since $\mathbf{E}$ is a PWF POVM, we have $\mathrm{Tr}\left(A_A^{\mathbf{u}}E_k\right) \geqslant 0$ for each $k$.
	Furthermore, because $\ket{k}\bra{k}$ is a stabilizer state, we have $1/(d_B)\mathrm{Tr}\left(A_B^{\mathbf{v}}\ket{k}\bra{k}\right) \geqslant 0$ for each $k$.
	Thus, it can be verified that the discrete Wigner function of $\mathcal{E}$ is positive, i.e., the measurement channel $\mathcal{E}$ is a CPWP channel.
	
	Conversely, suppose that the measurement channel $\mathcal{E}_{A \rightarrow B}$ is a CPWP channel.
	Then the discrete Wigner function of its Choi--Jamiołkowski matrix is positive, i.e.,
	\begin{equation}
		W_{\mathcal{E}}(\mathbf{v}|\mathbf{u}) = \displaystyle\sum_{k} W(E_k|\mathbf{u})W_{\ket{k}\bra{k}}(\mathbf{v}) \geqslant 0
	\end{equation}
	for all $\mathbf{u}$ and $\mathbf{v}$.
	Since $\{\ket{k}_B\}$ is a stabilizer basis, $\ket{k}\bra{k}_B$ is a pure stabilizer state for all $k$, the Wigner function of $\ket{k}\bra{k}_B$ is uniformly supported on an affine Lagrangian subspace $L_k$ \cite{W1987,G2006}, namely
	\begin{equation}
		W_{\ket{k}\bra{k}}(\mathbf{v}) =
		\begin{cases}
			\dfrac{1}{d_B}, & \mathbf{v}\in L_k,\\[4pt] 0, & \mathbf{v}\notin L_k.
		\end{cases}
	\end{equation}
	Moreover, different stabilizer basis states have disjoint supports, i.e.,
	\begin{equation}
		L_k\cap L_l=\varnothing,\quad k\neq l.
	\end{equation}
	
	Now fix arbitrary $k$ and $\mathbf{u}$, and choose any $\mathbf{v}_k\in L_k$, we have $W_{\ket{k}\bra{k}}(\mathbf{v}_k)=1/d_B$, and for any $l\neq k$, $W_{\ket{l}\bra{l}}(\mathbf{v}_k)=0$.
	Substituting $\mathbf{v}=\mathbf{v}_k$ into the expression of $W_{\mathcal{E}}(\mathbf{v}|\mathbf{u})$, we can get
	\begin{equation}\label{Lagrange affine}
			W_{\mathcal{E}}(\mathbf{v}_k|\mathbf{u}) = \sum_{l}W(E_l|\mathbf{u})W_{\ket{l}\bra{l}}(\mathbf{v}_k) = \frac{1}{d_B}W(E_k|\mathbf{u}).
	\end{equation}
	Since $\mathcal{E}_{A \rightarrow B}$ is a CPWP channel, we have $W_{\mathcal{E}}(\mathbf{v}_k|\mathbf{u})\geqslant 0$.
	Thus, $W(E_k|\mathbf{u})\geqslant 0$.
	Because $k$ and $\mathbf{u}$ are arbitrary, every effect $E_k$ has a positive discrete Wigner function.
	Therefore, $\mathbf{E}=\{E_k\}_{k=0}^{d_B-1}$ is a PWF POVM.
	
	This completes the proof.
\end{proof}
Thus, we can transform the question "How many magic resources at minimum do we need to add to a restricted measurement to achieve a global measurement" into the question "How many magic resources at minimum do we need to add to a free channel to achieve a given channel", which is consistent with the idea of the channel simulation task proposed in Ref. \cite{WL2026}. 
It gave the simulation cost of how many magic resources we need to simulate the target channel, and the definition of the simulation cost is defined as follows
\begin{definition}
	The simulation cost of simulating a target channel $\mathcal{M}$ by using a pure magic state is defined as follows
	\begin{equation}
		C_{\mathcal{F}}^{\varepsilon}(\mathcal{M};\omega) = \min \left\{n: \exists \mathcal{N} \in \mathcal{F}, \mathrm{s.t.} \frac{1}{2} ||\mathcal{M}-\mathcal{N}(\omega^{\otimes n}\otimes \cdot)||_{\diamond}\leqslant \varepsilon\right\},
	\end{equation}
	where $\mathcal{F}$ denotes the set of free operations, $\omega$ is a pure magic state, and we denote $C_{\mathcal{F}}(\mathcal{M};\omega)$ as exact simulation cost when $\varepsilon=0$.
\end{definition}

In this paper, we only consider the process of exact simulation.
Therefore, we can extend the results of Ref. \cite{WL2026} to qudit systems, i.e., give the simulation cost of using a given magic state and a CPWP channel to simulate the target channel.
Based on this idea, we can get the following theorem.

\begin{theorem}\label{theorem 2}
	Let $\mathcal{M}_{A\rightarrow B}$ be a quantum channel and $\omega$ be a pure magic state. 
	Suppose that there exist a number $\lambda \geqslant 1$ and two CPWP channels $\mathcal{K}$ and $\mathcal{L}$ such that $\mathcal{M} = \lambda \mathcal{K} - (\lambda-1)\mathcal{L}$. 
	Then the exact simulation cost of $\mathcal{M}$ satisfies
	\begin{equation}
		\begin{aligned}
			\frac{M(\mathcal{M})}{M(\omega)} &\leqslant C_{\mathcal{A}}(\mathcal{M};\omega) \\ &\leqslant \min\left\{n:\mathrm{Tr}(E\omega^{\otimes n})=1,0\leqslant W(E|\mathbf{u}) \leqslant \frac{1}{\lambda},  \forall\mathbf{u}\right\},
		\end{aligned}
	\end{equation}
	where $E$ is an effect such that $0 \leqslant E \leqslant \mathds{1}$, and we set the upper bound to $+\infty$ when no such effect $E$ exists.
	
\end{theorem}

\begin{proof}
	To begin with, we prove the lower bound. 
	Suppose that $n$ copies of $\omega$ can exactly simulate $\mathcal{M}$. 
	Then there exists a CPWP channel $\mathcal{N} \in \mathcal{A}$ such that
	\begin{equation}
		\mathcal{M}(\cdot) = \mathcal{N}\left(\omega^{\otimes n}\otimes\cdot\right).
	\end{equation}
	Next, we can get that
	\begin{equation}
			M(\mathcal{M}) \leqslant M (\omega^{\otimes n}) = nM(\omega),
	\end{equation}
	where the inequality follows from the monotonicity of the mana of quantum channels, the equality follows from the additivity of the mana of quantum states.
	Then we can get
	\begin{equation}
		\frac{M(\mathcal{M})}{M(\omega)} \leqslant n.
	\end{equation}
	Due to the definition of the exact simulation cost, we have
	\begin{equation}
		\frac{M(\mathcal{M})}{M(\omega)} \leqslant C_{\mathcal{A}}(\mathcal{M};\omega).
	\end{equation}
	
	We define a linear map $\mathcal{N}_{RA\rightarrow B}$ as follows:
	\begin{equation}\label{definition of N}
		\begin{aligned}
			\mathcal{N}_{RA\rightarrow B}(\rho_{RA}) =& \mathcal{M}_{A\rightarrow B}\Bigl(\mathrm{Tr}_R\left[(E_R\otimes\mathds{1}_A)\rho_{RA}\right]\Bigr) \\ &+ \mathcal{L}_{A\rightarrow B}\Bigl(\mathrm{Tr}_R\left[\bigl((\mathds{1}_R-E_R)\otimes\mathds{1}_A\bigr)\rho_{RA}\right]\Bigr).
		\end{aligned}
	\end{equation}
	Since $E$ is an effect, $\rho$ is a quantum state, it is easy to verify that $\mathcal{N}$ is a completely positive map.
	Next, since $\mathcal{M}$ and $\mathcal{L}$ are trace-preserving maps, we have
	\begin{equation}
		\begin{aligned}
			\mathrm{Tr}\left[\mathcal{N}_{RA\rightarrow B}(\rho_{RA})\right] &= \mathrm{Tr}\left[ (E_R\otimes\mathds{1}_A)\rho_{RA} \right] \\&+ \mathrm{Tr}\left[\bigl((\mathds{1}_R-E_R)\otimes\mathds{1}_A\bigr)\rho_{RA}\right]  \\ &= \mathrm{Tr}(\rho_{RA}).
		\end{aligned}
	\end{equation}
	Thus, $\mathcal{N}$ is a quantum channel.
	
	Subsequently, we prove that $\mathcal{N}$ is a CPWP channel.
	For arbitrary phase space point $\mathbf{u}_R$ and $\mathbf{u}_A$, we have
	\begin{equation}
		\begin{aligned}
			W_{\mathcal{N}}(\mathbf{v}_B|&\mathbf{u}_R,\mathbf{u}_A) = \frac{1}{d_B}\mathrm{Tr}\left[A_{\mathbf{v}_B}\mathcal{N}_{RA\rightarrow B}\left(A_{\mathbf{u}_R}\otimes A_{\mathbf{u}_A}\right)\right] \\ &~~~~~~~~~~~= \frac{1}{d_B}\mathrm{Tr}\left[A_{\mathbf{v}_B}\mathcal{M}_{A\rightarrow B}\left(\mathrm{Tr}_R\left[(E_R\otimes\mathds{1}_A)A_{\mathbf{u}_R}\otimes A_{\mathbf{u}_A}\right]\right)\right] \\ &+ \frac{1}{d_B}\mathrm{Tr}\left[A_{\mathbf{v}_B}\mathcal{L}_{A\rightarrow B}\left(\mathrm{Tr}_R\left[\bigl((\mathds{1}_R-E_R)\otimes\mathds{1}_A\bigr)A_{\mathbf{u}_R}\otimes A_{\mathbf{u}_A}\right]\right)\right] \\ &~~~~~~~~~~~= \frac{1}{d_B}\mathrm{Tr}\left[A_{\mathbf{v}_B}\mathcal{M}_{A\rightarrow B}\left(\mathrm{Tr}(E_RA_{\mathbf{u}_R})A_{\mathbf{u}_A}\right)\right] \\ &~~~~~~~~~~~+ \frac{1}{d_B}\mathrm{Tr}\left[A_{\mathbf{v}_B}\mathcal{L}_{A\rightarrow B}\left([1-\mathrm{Tr}(E_RA_{\mathbf{u}_R})]A_{\mathbf{u}_A}\right)\right] \\ &~~~~~~~~~~~= W(E|\mathbf{u}) W_{\mathcal{M}}(\mathbf{v}_B|\mathbf{u}_A) + [1-W(E|\mathbf{u})]W_{\mathcal{L}}(\mathbf{v}_B|\mathbf{u}_A),
		\end{aligned}
	\end{equation}
	where the first equality and last equality follow from the definition of the Wigner function of quantum channels, the second equality follows from Eq. \eqref{definition of N}.
	Since $\mathcal{M} = \lambda \mathcal{K} - (\lambda-1)\mathcal{L}$, we have 
	\begin{equation}
		\begin{aligned}
			W_{\mathcal{M}}(\mathbf{v}_B|\mathbf{u}_A) &= \frac{1}{d_B}\mathrm{Tr}\left[A_{\mathbf{v}_B}\mathcal{M}(A_{\mathbf{u}_A})\right] \\ &= \frac{1}{d_B}\mathrm{Tr}\left[A_{\mathbf{v}_B}\left(\lambda \mathcal{K}(A_{\mathbf{u}_A}) - (\lambda-1)\mathcal{L} (A_{\mathbf{u}_A})\right)\right] \\ &= \lambda W_{\mathcal{K}}(\mathbf{v}_B|\mathbf{u}_A)- (\lambda-1)W_{\mathcal{L}}(\mathbf{v}_B|\mathbf{u}_A).
		\end{aligned}
	\end{equation}
	Thus, we can get that
	\begin{equation}
		\begin{aligned}
			W_{\mathcal{N}}(\mathbf{v}_B|\mathbf{u}_R,\mathbf{u}_A) &= \lambda W(E|\mathbf{u}) W_{\mathcal{K}}(\mathbf{v}_B|\mathbf{u}_A) \\ &+ [1-\lambda W(E|\mathbf{u})]W_{\mathcal{L}}(\mathbf{v}_B|\mathbf{u}_A).
		\end{aligned}
	\end{equation}
	Therefore, we have $\lambda W(E|\mathbf{u}) \geqslant 0$ and $[1-\lambda W(E|\mathbf{u})] \geqslant 0$ since $0 \leqslant W(E|\mathbf{u}) \leqslant 1/\lambda$.
	Since $\mathcal{K}$ and $\mathcal{L}$ are CPWP channels, their Wigner functions are positive. 
	Hence,
	\begin{equation}
		W_{\mathcal{N}} (\mathbf{v}|\mathbf{u}_R,\mathbf{u}_A) \geqslant 0
	\end{equation}
	for all $\mathbf{u}_R$, $\mathbf{u}_A$, and $\mathbf{v}$. 
	Therefore, $\mathcal{N}$ is a CPWP channel.
	
	At last, since $\mathrm{Tr}\left[(\mathds{1}_R-E_R)\omega_R^{\otimes n}\right]=0$, then for every input state $\rho_A$, we have
	\begin{equation}
		\begin{aligned}
			\mathcal{N}_{RA\rightarrow B}\left(\omega_R^{\otimes n}\otimes\rho_A\right) &= \mathcal{M}_{A\rightarrow B}\left( \mathrm{Tr}_R\left[(E_R\otimes\mathds{1}_A)(\omega_R^{\otimes n}\otimes\rho_A)\right]\right) \\ &= \mathcal{M}_{A\rightarrow B}(\rho_A).
		\end{aligned}
	\end{equation}
	Therefore, $n$ copies of pure magic state $\omega$ can exactly simulate $\mathcal{M}$ by a CPWP channel, which implies
	\begin{equation}
		C_{\mathcal{A}}(\mathcal{M};\omega) \leqslant n.
	\end{equation}
	
	This completes the proof.
\end{proof}

Thus, we have obtained the upper and lower bounds on the simulation cost of simulating an arbitrary channel using CPWP channels and magic states. 
The channel simulation bound provides sufficient upper and lower bounds on the amount of magic resources needed to implement a global measurement.
Therefore, Theorem \ref{theorem 2} gives general lower and upper bounds on the exact simulation cost. 
Although these bounds may be loose for general target channels, the following proposition shows that they can coincide for a concrete non-CPWP measurement channel. 
This example also illustrates how many magic resources are required to simulate a measurement beyond the class of CSMs.

\begin{proposition}\label{proposition 1}
	Let $\mathbf{M}=\{M_0,M_1,M_2\}$ be a three-outcome POVM with $M_0=\ket{S}\bra{S}/2+\mathds{1}/6, M_1=\mathds{1}-M_0, M_2=0$, where $\ket{S} = (\ket{1}-\ket{2})/\sqrt{2}$ is the Strange state.
	Define $\mathcal{M}$ as the corresponding measurement channel associated to $\mathbf{M}$. 
	Then we can get
	\begin{equation}
		C_{\mathcal{A}}(\mathcal{M};\ket{S}\bra{S})=1,
	\end{equation}
	where $C_{\mathcal{A}}(\mathcal{M};\ket{S}\bra{S})$ is the simulation cost of exactly simulating $\mathcal{M}$ in terms of the Strange state $\ket{S}\bra{S}$.
\end{proposition}

\begin{proof}
	First, we prove that the lower bound of simulation cost is $C_{\mathcal{A}}(\mathcal{M};\ket{S}\bra{S}) \geqslant 1$.
	For the Strange state $\ket{S}\bra{S}$, Fig. \ref{fig1} gives its discrete Wigner function.
	\begin{figure}[H]
		\centering
		\includegraphics[height=5cm,width=5cm]{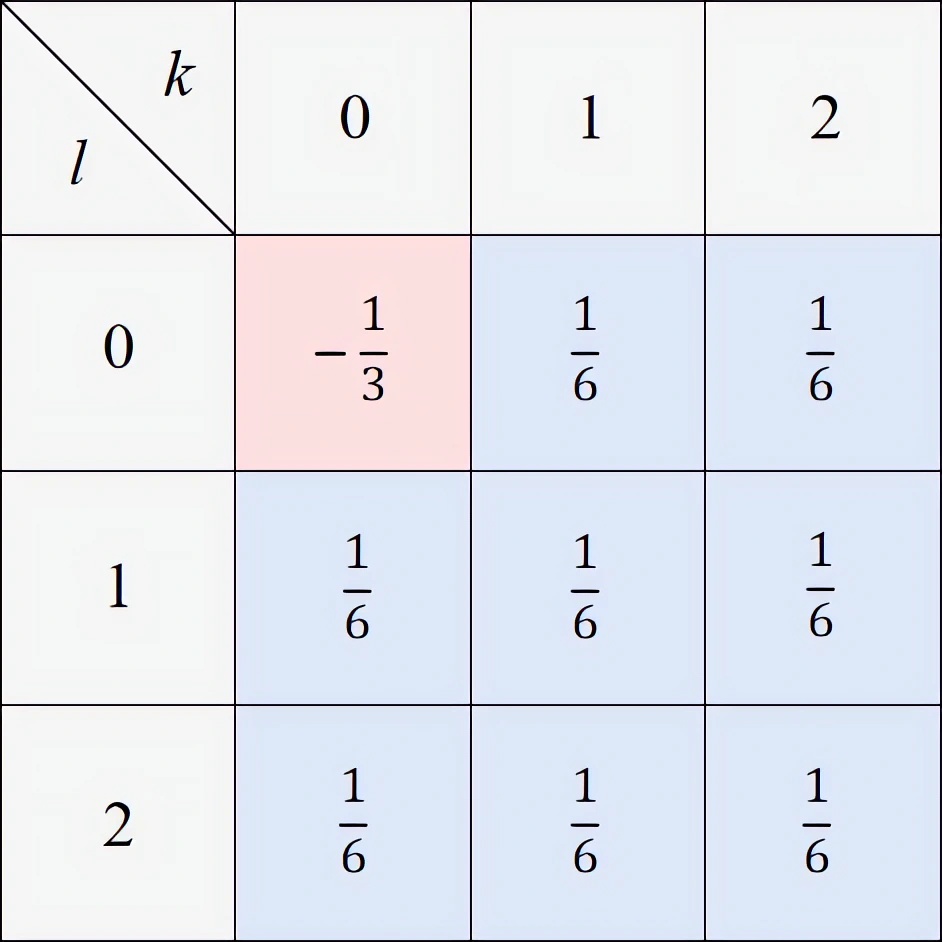}
		\caption{The discrete Wigner function of the Strange state.}
		\label{fig1}
	\end{figure}
	Let $\mathbf{u}_-$ denote the phase space point at which the discrete Wigner function of the Strange state takes a negative value, and let $\mathbf{u}_+$ denote the phase space point at which it takes a positive value, i.e.,
	\begin{equation}\label{Wigner function of Strange state}
		W_{\ket{S}\bra{S}}(\mathbf{u}_-)=-\frac{1}{3}, \quad
		W_{\ket{S}\bra{S}}(\mathbf{u}_+)=\frac{1}{6}.
	\end{equation}
	According to Eq. \eqref{definition of state mana}, we can get
	\begin{equation}
		M(\ket{S}\bra{S}) = \log_2\sum_{\mathbf{u}} \left|W_{\ket{S}\bra{S}}(\mathbf{u})\right| = \log_2\frac{5}{3}.
	\end{equation}
	
	Subsequently, we prove that $\mathbf{M}$ is not a PWF POVM. 
	By the definition of the discrete Wigner function of an effect, we have
	\begin{equation}\label{Wigner function of M0}
		\begin{aligned}
			W(M_0|\mathbf{u}_-) &= \mathrm{Tr}(A_{\mathbf{u}_-}M_0) \\ &= \frac{1}{2}\mathrm{Tr}(A_{\mathbf{u}_-}\ket{S}\bra{S}) + \frac{1}{6}\mathrm{Tr}\left(A_{\mathbf{u}_-}\mathds{1}\right) \\ &= \frac{3}{2}W_{\ket{S}\bra{S}}(\mathbf{u}_-)+\frac{1}{6} \\ &= -\frac{1}{3},
		\end{aligned}
	\end{equation}
	where the third equality follows from $\mathrm{Tr}(A_{\mathbf{u}}) = 1$ for each point $\mathbf{u}$ in phase space.
	Thus $\mathbf{M}$ is not a PWF POVM. 
	Then according to Theorem \ref{theorem 1}, the corresponding measurement channel $\mathcal{M}$ is not a CPWP channel.
	
	Next, we compute the mana of the measurement channel $\mathcal{M}$.
	For the measurement channel $\mathcal{M}$ associated with a POVM
	$\mathbf{M}=\{M_0,M_1,M_2\}$, we have
	\begin{equation}
			\sum_{\mathbf{v}} \left|W_{\mathcal{M}}(\mathbf{v}|\mathbf{u})\right|= \sum_{k=0}^{2} \sum_{\mathbf{v}\in L_k} \left| \frac{1}{3}W(M_k|\mathbf{u}) \right|= \sum_{k=0}^{2} \left|W(M_k|\mathbf{u})\right|,
	\end{equation}
	where the first equality follows from Eq. \eqref{Lagrange affine}, the last equality follows from the sum traverses all $\mathbf{v}$.
	Then we can get
	\begin{equation}
		W(M_0|\mathbf{u}_+) = \frac{3}{2}W_{\ket{S}\bra{S}} (\mathbf{u}_+)+\frac{1}{6} = \frac{5}{12}.
	\end{equation}
	Since $M_1=\mathds{1}-M_0$ and $M_2=0$, we have
	\begin{gather}
		W(M_1|\mathbf{u}_-) = \frac{4}{3}, \quad W(M_1|\mathbf{u}_+) = \frac{7}{12}, \\ 
		W(M_2|\mathbf{u}_-)=0, \quad W(M_2|\mathbf{u}_+)=0.
	\end{gather}
	Therefore, at the negative phase space point $\mathbf{u}_-$, we have
	\begin{equation}
		\sum_{k=0}^{2}\left|W(M_k|\mathbf{u}_-)\right| = \left|-\frac{1}{3}\right|+\left|\frac{4}{3}\right|+0=\frac{5}{3}.
	\end{equation}
	At the positive phase space point $\mathbf{u}_+$, we have
	\begin{equation}
		\sum_{k=0}^{2}\left|W(M_k|\mathbf{u}_+)\right| =
		\left|\frac{5}{12}\right|+\left|\frac{7}{12}\right|+0=1.
	\end{equation}
	Consequently,
	\begin{equation}
		M(\mathcal{M}) = \log_2\max_{\mathbf{u}}\sum_{\mathbf{v}} \left|W_{\mathcal{M}}(\mathbf{v}|\mathbf{u})\right| =  \log_2\frac{5}{3}.
	\end{equation}
	Then it can be obtained that
	\begin{equation}
		C_{\mathcal{A}}(\mathcal{M};\ket{S}\bra{S}) \geqslant \frac{M(\mathcal{M})}{M(\ket{S}\bra{S})} = \frac{\log_2(5/3)}{\log_2(5/3)} = 1.
	\end{equation}
	
	Afterwards, we prove that the upper bound of simulation cost is $C_{\mathcal{A}}(\mathcal{M};\ket{S}\bra{S}) \leqslant 1$.
	Let $\lambda=4/3$.
	Define a three-outcome POVM $\mathbf{E}=\{E_0,E_1,E_2\}$ with $E_0=\mathds{1}, E_1=0, E_2=0$.
	It is easy to see that $\mathbf{E}$ is a PWF POVM.
	Moreover, define another three-outcome POVM $\mathbf{F}=\{F_0,F_1,F_2\}$ with
	\begin{equation}
		F_i = \frac{M_i+(\lambda-1)E_i}{\lambda}, \quad i \in \left\{0,1,2\right\}.
	\end{equation}
	It is easy to see that $\mathbf{F}$ is a POVM since $F_0,F_1,F_2$ are positive semidefinite and $F_0+F_1+F_2 = \mathds{1}$. 
	We now verify that $\mathbf{F}$ is a PWF POVM. 
	By direct calculation, we have
	\begin{gather}
		W(F_0|\mathbf{u}_-) = 0, \quad W(F_1|\mathbf{u}_-) = 1, \quad W(F_2|\mathbf{u}_-) = 0, \\ W(F_0|\mathbf{u}_+) = \frac{9}{16}, \quad W(F_1|\mathbf{u}_+) = \frac{7}{16}, \quad W(F_2|\mathbf{u}_+)=0.
	\end{gather}
	Thus $\mathbf{F}$ is a PWF POVM. 
	Then let $\mathcal{F}$ and $\mathcal{E}$ be the measurement channels associated with $\mathbf{F}$ and $\mathbf{E}$, respectively. 
	According to Theorem \ref{theorem 1}, both $\mathcal{F}$ and $\mathcal{E}$ are CPWP channels. 
	Furthermore, we can get that
	\begin{equation}
		\mathcal{M} = \lambda\mathcal{F}-(\lambda-1)\mathcal{E}.
	\end{equation}

	Then we choose an effect $G$ such that
	\begin{equation}
		G = \frac{\mathds{1}+\ket{S}\bra{S}}{2}.
	\end{equation}
	It can be verified that  $0\leqslant G\leqslant\mathds{1}$ and $\mathrm{Tr}\left(G\ket{S}\bra{S}\right) = 1$.
	Moreover, we have
	\begin{gather}
		W(G|\mathbf{u}_-) = \frac{1}{2}+\frac{3}{2}W_{\ket{S}\bra{S}}(\mathbf{u}_-) = 0, \\ W(G|\mathbf{u}_+) = \frac{1}{2}+\frac{3}{2}W_{\ket{S}\bra{S}}(\mathbf{u}_+) = \frac{3}{4} = \frac{1}{\lambda}.
	\end{gather}
	Therefore, we have $0\leqslant W(G|\mathbf{u}) \leqslant 1/\lambda$ for arbitrary phase space point $\mathbf{u}$.
	The upper bound in Theorem \ref{theorem 2} applies with $n=1$, and hence
	\begin{equation}
	C_{\mathcal{A}}(\mathcal{M};\ket{S}\bra{S}) \leqslant 1.
	\end{equation}
	Combining the lower and upper bounds gives
	\begin{equation}
	C_{\mathcal{A}}(\mathcal{M};\ket{S}\bra{S}) = 1.
	\end{equation}
	
	This completes the proof.
\end{proof}

Proposition \ref{proposition 1} shows that the lower and upper bounds in Theorem \ref{theorem 2} coincide for the particular non-CPWP measurement channel considered above. 
Thus, its exact channel simulation cost is completely characterized. 
Moreover, the following proposition proves that one copy of the Strange state strictly improves the optimal success probability by using CSMs.

\begin{proposition}\label{proposition 2}
	Let $\rho_0 = \ket{S}\bra{S}$, $\rho_1 = (\mathds{1}-\ket{S}\bra{S})/2$.
	Consider the ensemble  $\Pi = \{(1/3,\rho_0), 
	(2/3,\rho_1)\}$, then the optimal success probability of discriminating this ensemble assisted by a copy of Strange state is strictly bigger than the optimal success probability using PWF POVMs, i.e.,
	\begin{equation}
		p_{\mathrm{succ}}^{\mathbf{PWF}+\ket{S}\bra{S}}(\Pi) > p_{\mathrm{succ}}^{\mathbf{PWF}}(\Pi).
	\end{equation}
\end{proposition}

\begin{proof}
	Since the two quantum states $\rho_0$ and $\rho_1$ have orthogonal supports, they can be perfectly distinguished by an unrestricted measurement. 
	Thus, we can use unrestricted global POVMs to distinguish them perfectly, i.e., we have
	\begin{equation}
		p_{\mathrm{succ}}^{\mathbf{GLOBAL}}(\Pi)=1.
	\end{equation}
	
	Next, consider an arbitrary binary PWF POVM $\mathbf{E} = \{E,\mathds{1}-E\}$, where the outcome corresponding to $E$ is used to guess $\rho_0$.
	The optimal success probability achieved by this POVM is
	\begin{equation}
		\begin{aligned}
			p_{\mathrm{succ}}^{\mathbf{E}}(\Pi) &= \frac{1}{3}\mathrm{Tr}(E\rho_0) + \frac{2}{3}\mathrm{Tr}[(\mathds{1}-E)\rho_1]  \\ &= \frac{2}{3} + \mathrm{Tr} \left[E\left(\frac{1}{3}\rho_0 - \frac{2}{3}\rho_1\right)\right].
		\end{aligned}
	\end{equation}

	According to Eq. \eqref{Wigner function of Strange state}, we have
	\begin{equation}
		\mathrm{Tr}\left(A_{\mathbf{u}_{-}}\ket{S}\bra{S}\right) = -1.
	\end{equation}
	Since the qutrit phase-space point operator $A_{\mathbf{u}_{-}}$ is a unitary operator with eigenvalue $\{+1,+1,-1\}$ \cite{ZLZ+2024}, the above equality implies that $\ket{S}\bra{S}$ is the eigenvector of $A_{\mathbf{u}_{-}}$ corresponding to the eigenvalue $-1$.
	Therefore, we have
	\begin{equation}
		A_{\mathbf{u}_{-}} = (-1)\ket{S}\bra{S} + 1(\mathds{1}-\ket{S}\bra{S}) = \mathds{1}-2\ket{S}\bra{S},
	\end{equation}
	where the first equality follows from $\mathds{1}-\ket{S}\bra{S}$ corresponds to the eigenvector with the eigenvalue of $+1$.
	Then we can get
	\begin{equation}
		\frac{1}{3}\rho_0-\frac{2}{3}\rho_1 = \frac{2}{3}\ket{S}\bra{S}-\frac{1}{3}\mathds{1} = -\frac{1}{3}A_{\mathbf{u}_-}.
	\end{equation}
	It follows that
	\begin{equation}
		\begin{aligned}
			p_{\mathrm{succ}}^{\mathbf{E}}(\Pi) &= \frac{2}{3}+\mathrm{Tr}\left[E\left(\frac{1}{3}\rho_0 - \frac{2}{3}\rho_1\right)\right] \\ &= \frac{2}{3} - \frac{1}{3}\mathrm{Tr}(EA_{\mathbf{u}_-}) \\ &= \frac{2}{3}-\frac{1}{3}W(E|\mathbf{u}_-).
		\end{aligned}
	\end{equation}
	Since $E$ is a PWF effect, we have $W(E|\mathbf{u}_-) \geqslant 0$.
	Then we can get
	\begin{equation}
		p_{\mathrm{succ}}^{\mathbf{PWF}}(\Pi) \leqslant \frac{2}{3}.
	\end{equation}
	
	Now we consider the POVM $\mathbf{M}=\{M_0,M_1\}$ with $M_0=\ket{S}\bra{S}/2+\mathds{1}/6, M_1=\mathds{1}-M_0$.
	If the outcome corresponding to $M_0$ is used to guess $\rho_0$, while the outcome $M_1$ is used to guess $\rho_1$, then we have
	\begin{equation}
		\begin{aligned}
			p_{\mathrm{succ}}^{\mathbf{M}}(\Pi) &= \frac{2}{3}+\mathrm{Tr}\left[M_0\left(\frac{1}{3}\rho_0 - \frac{2}{3}\rho_1\right)\right] \\ &= \frac{2}{3} - \frac{1}{3}\mathrm{Tr}(M_0A_{\mathbf{u}_-}) \\ &= \frac{2}{3}-\frac{1}{3}W(M_0|\mathbf{u}_-) \\ &= \frac{7}{9},
		\end{aligned}
	\end{equation}
	where the last equality follows from Eq. \eqref{Wigner function of M0}.
	At this time, we can add a zero effect to POVMs $\mathbf{E}$ and $\mathbf{M}$, respectively. 
	In this way, both $\mathbf{E} = \{E,\mathds{1}-E,0\}$ and $\mathbf{M} = \{M_0,M_1,0\}$ can be regarded as three-outcome POVMs.
	By Proposition \ref{proposition 1}, the measurement channel associated with the three-outcome POVM $\mathbf{M}$ can be exactly simulated by a CPWP channel using one copy of the Strange state.
	Equivalently, one copy of the Strange state allows CSMs to implement the same discrimination power as $\mathbf{M}$.
	Consequently,
	\begin{equation}
		p_{\mathrm{succ}}^{\mathbf{PWF}+\ket{S}\bra{S}}(\Pi) \geqslant \frac{7}{9}.
	\end{equation}
	Finally, we have
	\begin{equation}
		p_{\mathrm{succ}}^{\mathbf{PWF}+\ket{S}\bra{S}}(\Pi) \geqslant \frac{7}{9} > \frac{2}{3} \geqslant p_{\mathrm{succ}}^{\mathbf{PWF}}(\Pi).
	\end{equation}
	
	This completes the proof.
\end{proof}

Proposition \ref{proposition 2} proves that consumable magic resources can improve the success probability under CSMs. 
We further show that three copies of the Strange state are necessary and sufficient for perfectly discriminating the ensemble $\Pi$ using a PWF POVM. 

\begin{proposition}\label{proposition 3}
	Let $\rho_0=\ket{S}\bra{S}$ and $\rho_1=(\mathds{1}-\ket{S}\bra{S})/2$. 
	Consider the ensemble $\Pi=\{(1/3,\rho_0),(2/3,\rho_1)\}$. 
	Then
	\begin{equation}
		p_{\mathrm{succ}}^{\mathbf{PWF}+(\ket{S}\bra{S})^{\otimes 3}}(\Pi)=1.
	\end{equation}
	Moreover, three copies of the Strange state are necessary for achieving perfect discrimination under PWF POVMs, i.e., $E_c^{\mathbf{PWF}}=3$ for $\Pi$ with $\omega=\ket{S}\bra{S}$.
\end{proposition}

\begin{proof}
	First, we prove that at least three copies of the Strange state are required to perfectly discriminate the ensemble $\Pi$ by using PWF POVMs. 
	Consider an arbitrary binary PWF POVM $\mathbf{E}=\{E,\mathds{1}-E\}$ acting on the joint system consisting of the input qutrit and $n$ copies of $\ket{S}\bra{S}$, where the outcome corresponding to $E$ is used to guess $\rho_0$. 
	Since $\rho_0/3-2\rho_1/3=-A_{\mathbf{u}_-}/3$, we have
	\begin{equation}
		p_{\mathrm{succ}}^{\mathbf{E}}\left[\Pi,(\ket{S}\bra{S})^{\otimes n}\right]=\frac{2}{3}-\frac{1}{3}\mathrm{Tr}\left[E(A_{\mathbf{u}_-}\otimes(\ket{S}\bra{S})^{\otimes n})\right].
	\end{equation}
	Let $\mathbf{v}$ denote a phase space point of the auxiliary system. 
	According to the discrete Wigner representation of $(\ket{S}\bra{S})^{\otimes n}$, we have
	\begin{equation}
		\begin{aligned}
			\mathrm{Tr}\left[E(A_{\mathbf{u}_-}\otimes(\ket{S}\bra{S})^{\otimes n})\right]&=\sum_{\mathbf{v}}W_{(\ket{S}\bra{S})^{\otimes n}}(\mathbf{v})W(E|\mathbf{u}_-,\mathbf{v}) \\&\geqslant\sum_{\mathbf{v}:W_{(\ket{S}\bra{S})^{\otimes n}}(\mathbf{v})<0}W_{(\ket{S}\bra{S})^{\otimes n}}(\mathbf{v}) \\&=\frac{1}{2}\left[1-\left(\frac{5}{3}\right)^n\right],
		\end{aligned}
	\end{equation}
	where the inequality follows from $0\leqslant W(E|\mathbf{u}_-,\mathbf{v})\leqslant1$ for all $\mathbf{v}$, and the last equality follows from $\sum_{\mathbf{v}}W_{(\ket{S}\bra{S})^{\otimes n}}(\mathbf{v})=1$ and the additivity of the mana of quantum states. Therefore,
	\begin{equation}\label{Strange assisted success upper bound}
		p_{\mathrm{succ}}^{\mathbf{PWF}+(\ket{S}\bra{S})^{\otimes n}}(\Pi)\leqslant \frac{1}{2}+\frac{1}{6}\left(\frac{5}{3}\right)^n.
	\end{equation}
	For $n\leqslant2$, the right-hand side is no greater than $26/27<1$. 
	Thus, $E_c^{\mathbf{PWF}}\geqslant3$.
	
	Subsequently, we prove that three copies of the Strange state are sufficient. 
	Define an effect on a single qutrit system as follows
	\begin{equation}
		G=\frac{\mathds{1}+\ket{S}\bra{S}}{2}.
	\end{equation}
	It is easy to verify that $0\leqslant G\leqslant\mathds{1}$ and $\mathrm{Tr}(G\ket{S}\bra{S})=1$. 
	According to Eq. \eqref{Wigner function of Strange state}, we have
	\begin{equation}
		W(G|\mathbf{u}_-)=0,\quad W(G|\mathbf{u}_+)=\frac{3}{4}.
	\end{equation}
	Since the discrete Wigner function of an effect is multiplicative under tensor products, we have
	\begin{equation}
		0\leqslant W(G^{\otimes 3}|\mathbf{v})\leqslant\left(\frac{3}{4}\right)^3=\frac{27}{64},
	\end{equation}
	where $\mathbf{v}=(\mathbf{v}_1,\mathbf{v}_2,\mathbf{v}_3)$ denotes an arbitrary phase space point of the three-qutrit auxiliary system. 
	Moreover,
	\begin{equation}
		\mathrm{Tr}\left[G^{\otimes 3}(\ket{S}\bra{S})^{\otimes 3}\right]=1.
	\end{equation}
	Now we construct a binary POVM $\mathbf{F}=\{F_0,F_1\}$ on the four-qutrit system with
	\begin{equation}\label{perfect Strange assisted POVM}
		F_1=(\mathds{1}-\omega)\otimes G^{\otimes 3},\quad F_0=\mathds{1}-F_1,
	\end{equation}
	where the identity in $F_0$ acts on the four-qutrit system. Since $0\leqslant\mathds{1}-\ket{S}\bra{S}\leqslant\mathds{1}$ and $0\leqslant G^{\otimes 3}\leqslant\mathds{1}$, we have $0\leqslant F_1\leqslant\mathds{1}$. 
	Hence, $\mathbf{F}$ is a POVM.
	Next, we verify that $\mathbf{F}$ is a PWF POVM. 
	By direct calculation, we have
	\begin{equation}
		W(\mathds{1}-\ket{S}\bra{S}|\mathbf{u}_-)=2, \quad W(\mathds{1}-\ket{S}\bra{S}|\mathbf{u}_+)=\frac{1}{2}.
	\end{equation}
	Therefore, for arbitrary phase space points $\mathbf{u}$ and $\mathbf{v}$, we have
	\begin{equation}
		0\leqslant W(F_1|\mathbf{u},\mathbf{v})=W(\mathds{1}-\ket{S}\bra{S}|\mathbf{u})W(G^{\otimes 3}|\mathbf{v})\leqslant\frac{27}{32}<1.
	\end{equation}
	Since $F_0=\mathds{1}-F_1$, it follows that
	\begin{equation}
		W(F_0|\mathbf{u},\mathbf{v})=1-W(F_1|\mathbf{u},\mathbf{v})\geqslant\frac{5}{32}>0.
	\end{equation}
	Thus, $\mathbf{F}$ is a PWF POVM. 
	Finally, we calculate the success probability achieved by $\mathbf{F}$. 
	Since $\rho_0=\ket{S}\bra{S}$ and $\rho_1=(\mathds{1}-\ket{S}\bra{S})/2$, we have
	\begin{equation}
		\mathrm{Tr}\left[F_1[\rho_0\otimes(\ket{S}\bra{S})^{\otimes 3}]\right]=0, \quad \mathrm{Tr}\left[F_1[\rho_1\otimes(\ket{S}\bra{S})^{\otimes 3}]\right]=1.
	\end{equation}
	Consequently, $\mathrm{Tr}\left[F_0[\rho_0\otimes(\ket{S}\bra{S})^{\otimes 3}]\right]=1$, and hence
	\begin{equation}
		\begin{aligned}
			p_{\mathrm{succ}}^{\mathbf{F}}\left[\Pi,(\ket{S}\bra{S})^{\otimes 3}\right]&=\frac{1}{3}\mathrm{Tr}\left[F_0[\rho_0\otimes(\ket{S}\bra{S})^{\otimes 3}]\right]\\
			&\quad+\frac{2}{3}\mathrm{Tr}\left[F_1[\rho_1\otimes(\ket{S}\bra{S})^{\otimes 3}]\right]\\
			&=1.
		\end{aligned}
	\end{equation}
	Thus, $p_{\mathrm{succ}}^{\mathbf{PWF}+\omega^{\otimes 3}}(\Pi)=1$ and $E_c^{\mathbf{PWF}}\leqslant3$. 
	Combining the lower and upper bounds gives $E_c^{\mathbf{PWF}}=3$. 
	This completes the proof. 
\end{proof}

Proposition \ref{proposition 3} determines the exact magic cost of optimal PWF discrimination for the ensemble $\Pi$. 
Moreover, taking $n=1$ in Eq. \eqref{Strange assisted success upper bound} and combining it with Proposition \ref{proposition 2}, we obtain $p_{\mathrm{succ}}^{\mathbf{PWF}+\ket{S}\bra{S}}(\Pi)=7/9$. 
These results show that consumable magic resources can not only strictly improve the discrimination power of CSMs but also enable perfect discrimination with a finite number of copies.
Based on Propositions \ref{proposition 2} and \ref{proposition 3},
we can further determine the optimal success probability of
discriminating the ensemble $\Pi$ assisted by an arbitrary number
of copies of the Strange state. 
Specifically, we give the following corollary.

\begin{corollary}\label{corollary Strange assisted discrimination}
	Let $\rho_0=\ket{S}\bra{S}$ and $\rho_1=(\mathds{1}-\ket{S}\bra{S})/2$. 
	Consider the ensemble $\Pi=\{(1/3,\rho_0),(2/3,\rho_1)\}$. 
	Then, for any nonnegative integer $n$, the optimal success probability of discriminating the ensemble $\Pi$ using PWF POVMs assisted by $n$ copies of the Strange state is given by
	\begin{equation}\label{exact Strange assisted success probability}
		p_{\mathrm{succ}}^{\mathbf{PWF}+(\ket{S}\bra{S})^{\otimes n}}(\Pi)=\min\left\{1,\frac{1}{2}+\frac{1}{6}\left(\frac{5}{3}\right)^n\right\},
	\end{equation}
	where $n=0$ corresponds to the case without magic resources assistance.
\end{corollary}

\begin{proof}
	According to Eq. \eqref{Strange assisted success upper bound} and the fact that the success probability of any discrimination strategy cannot exceed one, we have
	\begin{equation}\label{Strange assisted success upper bound corollary}
		p_{\mathrm{succ}}^{\mathbf{PWF}+(\ket{S}\bra{S})^{\otimes n}}(\Pi)\leqslant\min\left\{1,\frac{1}{2}+\frac{1}{6}\left(\frac{5}{3}\right)^n\right\}.
	\end{equation}
	
	When $n=0$, the success probability $2/3$ can be achieved by always guessing that the input state is $\rho_1$. 
	Combining this observation with Eq. \eqref{Strange assisted success upper bound corollary}, we have
	\begin{equation}
		p_{\mathrm{succ}}^{\mathbf{PWF}}(\Pi)=\frac{2}{3}.
	\end{equation}
	
	When $n=1$, according to Proposition \ref{proposition 2}, we have
	\begin{equation}
		p_{\mathrm{succ}}^{\mathbf{PWF}+\ket{S}\bra{S}}(\Pi) \geqslant\frac{7}{9}.
	\end{equation}
	Combining this inequality with Eq. \eqref{Strange assisted success upper bound corollary}, we obtain
	\begin{equation}
		p_{\mathrm{succ}}^{\mathbf{PWF}+\ket{S}\bra{S}}(\Pi)=\frac{7}{9}.
	\end{equation}
	
	Subsequently, we prove that the upper bound in Eq. \eqref{Strange assisted success upper bound corollary} can also be achieved when $n=2$. 
	For simplicity, let $P=\ket{S}\bra{S}$ and $Q=\mathds{1}-P$. 
	We construct a binary POVM $\mathbf{E}=\{E,\mathds{1}-E\}$ on the three-qutrit system, where
	\begin{equation}\label{two Strange assisted POVM}
		\begin{aligned}
			E={}&\frac{8}{9}P\otimes P\otimes P+\frac{1}{9}P\otimes Q\otimes Q\\&+\frac{1}{9}Q\otimes P\otimes Q+\frac{1}{9}Q\otimes Q\otimes P\\&+\frac{5}{18}Q\otimes Q\otimes Q.
		\end{aligned}
	\end{equation}
	Since $P$ and $Q$ are mutually orthogonal projectors, the five tensor product operators in Eq. \eqref{two Strange assisted POVM} are also mutually orthogonal projectors. 
	Moreover, since their coefficients lie in the interval $[0,1]$, we have $0\leqslant E\leqslant\mathds{1}$. 
	Thus, $\mathbf{E}$ is a POVM.
	Next, we verify that $\mathbf{E}$ is a PWF POVM.
	According to Eq. \eqref{Wigner function of Strange state},
	we have
	\begin{equation}
		\begin{aligned}
			W(P|\mathbf{u}_-)&=-1,\quad W(P|\mathbf{u}_+)=\frac{1}{2},\\ W(Q|\mathbf{u}_-)&=2,\quad W(Q|\mathbf{u}_+)=\frac{1}{2}.
		\end{aligned}
	\end{equation}
	Let $\mathbf{v}=(\mathbf{v}_1,\mathbf{v}_2,\mathbf{v}_3)$ denote an arbitrary phase space point of the three-qutrit system, and let $k$ denote the number of components of $\mathbf{v}$ equal to $\mathbf{u}_-$, i.e.,
	\begin{equation}
		k=\left|\left\{i\in\{1,2,3\}:\mathbf{v}_i=\mathbf{u}_-\right\}\right|.
	\end{equation}
	Since the discrete Wigner function of an effect is multiplicative under tensor products, by direct calculation, we have
	\begin{equation}
		W(E|\mathbf{v})=
		\begin{cases}
			3/16, & k=0,\\ 0, & k=1,\\ 1, & k=2,\\ 0, & k=3.
		\end{cases}
	\end{equation}
	Therefore, for every phase space point $\mathbf{v}$, we have $0\leqslant W(E|\mathbf{v})\leqslant 1$. 
	Since
	\begin{equation}
		W(\mathds{1}-E|\mathbf{v})=1-W(E|\mathbf{v})\geqslant0,
	\end{equation}
	$\mathbf{E}$ is a PWF POVM.
	
	Finally, we calculate the success probability achieved by $\mathbf{E}$ with the assistance of two copies of the Strange state.
	If the outcome corresponding to $E$ is used to guess $\rho_0$, while the outcome corresponding to $\mathds{1}-E$ is used to guess $\rho_1$, then, since $\rho_0=P$, $\rho_1=Q/2$, and $PQ=QP=0$, we have
	\begin{equation}
		\mathrm{Tr}\left[E(\rho_0\otimes P^{\otimes2})\right]=\frac{8}{9}, \quad \mathrm{Tr}\left[E(\rho_1\otimes P^{\otimes2})\right]=0.
	\end{equation}
	Consequently,
	\begin{equation}
		\begin{aligned}
			p_{\mathrm{succ}}^{\mathbf{E}}\left[\Pi,P^{\otimes2}\right]&=\frac{1}{3}\mathrm{Tr}\left[E(\rho_0\otimes P^{\otimes2})\right]\\&\quad+\frac{2}{3}\mathrm{Tr}\left[(\mathds{1}-E)(\rho_1\otimes P^{\otimes2})\right]\\&=\frac{1}{3}\times\frac{8}{9}+\frac{2}{3}\\&=\frac{26}{27}.
		\end{aligned}
	\end{equation}
	Combining this equality with Eq. \eqref{Strange assisted success upper bound corollary}, we obtain
	\begin{equation}
		p_{\mathrm{succ}}^{\mathbf{PWF}+(\ket{S}\bra{S})^{\otimes2}}(\Pi)=\frac{26}{27}.
	\end{equation}
	
	When $n\geqslant3$, according to Proposition \ref{proposition 3}, three copies of the Strange state are sufficient to perfectly discriminate the ensemble $\Pi$ using a PWF POVM. 
	Since the additional copies of the Strange state can be
	ignored, we have
	\begin{equation}
		p_{\mathrm{succ}}^{\mathbf{PWF}+(\ket{S}\bra{S})^{\otimes n}}(\Pi)=1,\quad n\geqslant3.
	\end{equation}
	
	Combining the above results, for any nonnegative integer $n$, we have
	\begin{equation}
		p_{\mathrm{succ}}^{\mathbf{PWF}+(\ket{S}\bra{S})^{\otimes n}}(\Pi)=\min\left\{1,\frac{1}{2}+\frac{1}{6}\left(\frac{5}{3}\right)^n\right\}.
	\end{equation}
	This completes the proof.
\end{proof}

\begin{figure*}
	\centering
	\includegraphics[height=7.5cm,width=9.5cm]{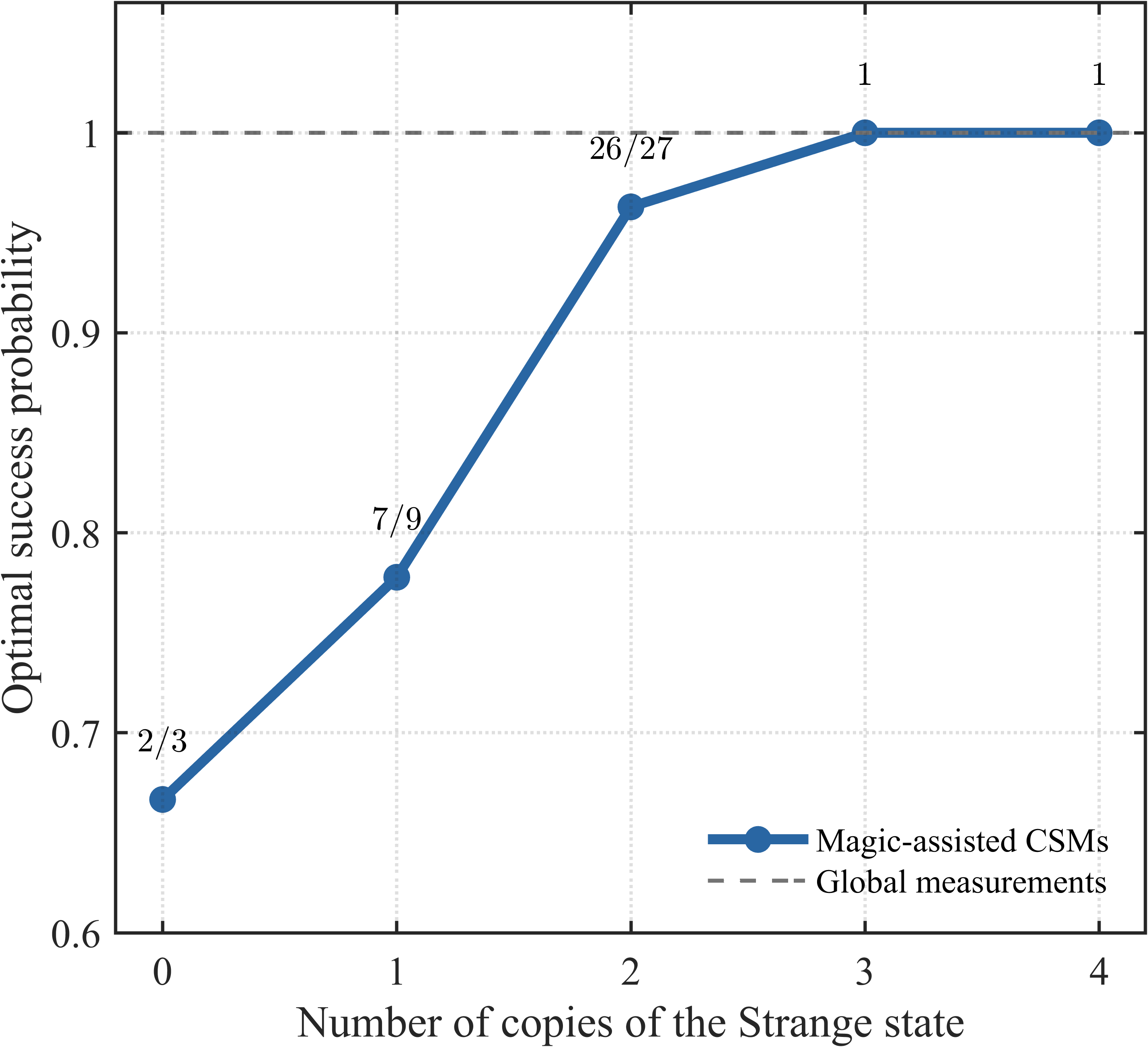}
	\caption{The optimal success probability of discriminating the ensemble $\Pi$ using CSMs assisted by $n$ copies of the Strange state. The blue solid line represents the optimal success probability achieved by magic-assisted PWF POVMs, while the gray dashed line represents the optimal success probability achieved by unrestricted global measurements.}
	\label{figure Strange assisted discrimination}
\end{figure*}

Corollary \ref{corollary Strange assisted discrimination} provides an exact characterization of the relationship between the number of copies of the Strange state and the optimal success probability of discriminating the ensemble $\Pi$ using CSMs. 
As shown in Fig. \ref{figure Strange assisted discrimination}, the optimal success probability increases from $2/3$ without magic assistance to $7/9$ and $26/27$ with one and two copies of the Strange state, respectively, and reaches the optimal success probability achieved by global measurements with three copies. 
This result not only exactly determines the magic cost of optimal discrimination for this ensemble but also quantitatively demonstrates how consumable magic resources enhance the discrimination power of CSMs.

However, for a general ensemble, it is difficult to analytically determine the magic cost of optimal discrimination. 
This motivates us to formulate the magic-assisted binary QSD problem as an SDP, which provides a direct method for computing the optimal success probability assisted by a given number of copies of a magic state.
We consider a binary QSD task with equal prior probabilities by using CSMs, i.e., using a PWF POVM $\mathbf{E}$ to distinguish the pair of states $\{\rho_0, \rho_1\}$ with the same prior probabilities.
Our goal is to maximize the average success probability, which is given by
\begin{equation}
	\begin{aligned}
		p_{\rm{succ}}^{\mathbf{E}}\left(\left\{1/2, \rho_i\right\}_{i=0}^{1}, \omega^{\otimes n}\right) = \max_{\{E_i\}\in\mathbf{E}}\frac{1}{2}\mathrm{Tr}\left[E_0 \left(\rho_0 \otimes \omega^{\otimes n}\right)\right] \\ + \frac{1}{2}\mathrm{Tr}\left[E_1 \left(\rho_1 \otimes \omega^{\otimes n}\right)\right].
	\end{aligned}
\end{equation}
Then, we can give the primal problem of SDP of the magic-assisted average
success probability by PWF POVMs.
\begin{equation}
	\begin{aligned}
		&\max_{E} \quad &&\frac{1}{2}+\frac{1}{2}\mathrm{Tr}\{E[(\rho_0 - \rho_1)\otimes \omega^{\otimes n}]\}, \\ &\mathrm{subject~to} \quad && 0\leqslant E\leqslant\mathds{1},  \\ &&& 0\leqslant \mathrm{Tr}(EA_{\mathbf{u}})\leqslant 1 \quad \forall u,
	\end{aligned}
\end{equation}
where $E$ is the effect associated with the guess $\rho_0$, so that $\{E,\mathds{1}-E\}$ is a binary PWF POVM, $\omega$ is a magic state, the constraints follow from the properties of PWF POVM.
We give the dual problem of this SDP in Appendix A.
This SDP provides a method to directly compute the optimal success probability of magic-assisted binary state discrimination for a fixed number of copies of magic states. 
By solving this SDP for different $n$, we can numerically estimate the minimum amount of magic resources required for PWF POVMs to achieve the performance of global measurements.

\section{Using quantum catalysts and memories to improve the discrimination power}
We have verified that adding magic resources to the set of CSMs can improve its discrimination power. 
Furthermore, a natural question is whether quantum catalysts and memories can be used to improve the discrimination power of CSMs. 
We will explore this question in this section.

The quantum catalyst is a quantum system that enables quantum state transformation that would otherwise be impossible. 
The core idea is that the catalytic process will not change the catalytic state, so the corresponding catalytic state can be reused to realize the transformation process of the same quantum state $\rho$.
Specifically, we consider a binary QSD task. 
In this setting, the goal of using a quantum catalyst to assist CSMs in the task of QSD is to find a quantum catalytic state $\mu_C$ and a CPWP operation $\Lambda$ such that
\begin{gather}
	\Lambda \left( \rho_i^A \otimes \mu_C \right) = \sigma_i^A \otimes \mu_C, \\ p_{\rm{succ}}^{\mathbf{PWF}}\left(\{p_i,\sigma_i\}_{i=0}^1\right) \geqslant p_{\rm{succ}}^{\mathbf{PWF}}\left(\{p_i,\rho_i\}_{i=0}^1\right),
\end{gather}
where $\mu_C$ is the unchanged catalytic state, which enables $\rho_i^A$ to be catalytically transformed into $\sigma_i^A$, and makes the new pair of states $\{p_i,\sigma_i\}_{i=0}^1$ more distinguishable than $\{p_i,\rho_i\}_{i=0}^1$ under CSMs.
Furthermore, the quantum memory is an auxiliary quantum system that can be updated during the evolution process and retain historical information in this paper, thereby influencing the inputs, operations, or measurement strategies in subsequent rounds \cite{CDP2008}. 
Unlike a quantum catalyst, a quantum memory is not required to return to its initial state, instead, it evolves throughout the process and carries historical information.

Inspired by the general state discrimination framework under LOCC protocol in Ref. \cite{PS2025}, we give the general state discrimination framework under CPWP protocol to study whether quantum catalysts and memories can be used to improve the discrimination power of CSMs.
First, we introduce the general state discrimination framework under CPWP protocol assisted by quantum memories.

Let $\rho_0,\rho_1$ be two states on system $A$ with prior probabilities $p_0,p_1>0$ and $p_0+p_1=1$. 
In the $j$-th round, let $Z_j\in\{0,1\}$ be an independent identically distributed (i.i.d.) random variable which is sampled independently according to
\begin{equation}
	\Pr(Z_j=i)=p_i,
\end{equation}
and the corresponding state $\rho_{Z_j}$ is prepared on system $A$. 
A memory-assisted PWF discrimination protocol consists of an initial memory state $\mu_C$, and a sequence of adaptive CPWP instruments $\Lambda_{AC}$. 
A quantum instrument is a completely positive and trace-non-increasing map that gives both classical and quantum outputs \cite{DL1970,KW2020}.
The instrument used in the $j$-th round depends on all previous classical outcomes of the protocol. 
It produces a classical guess $Y_j\in\{0,1\}$ and an updated memory state for the next round.
Subsequently, we can define the variable $X_j$ to indicate whether the classical guess $Y_j$ is the same as the original variable $Z_j$ and the total number of correct guesses $S_n$ after $n$ rounds 
\begin{gather}
	X_j=
	\begin{cases}
		1, & Y_j=Z_j,\\
		0, & Y_j\neq Z_j,
	\end{cases} \\
	S_n\coloneqq \sum_{j=1}^{n}X_j.
\end{gather}

Since catalyst-assisted discrimination is a special case of memory-assisted discrimination, we replace the evolving memory state $\mu_C$ with an unchanged catalytic state $\mu_C$ in the catalyst-assisted discrimination framework. 
Meanwhile, $\Lambda_{AC}$ is no longer an adaptive CPWP instrument, instead, it outputs a classical guess together with the unchanged catalytic state $\mu_C$.
Fig. \ref{fig2} describes the specific process of general state discrimination framework under CPWP protocol assisted by quantum memories and catalysts.
\begin{figure*}
	\centering
	\includegraphics[height=6cm,width=16.5cm]{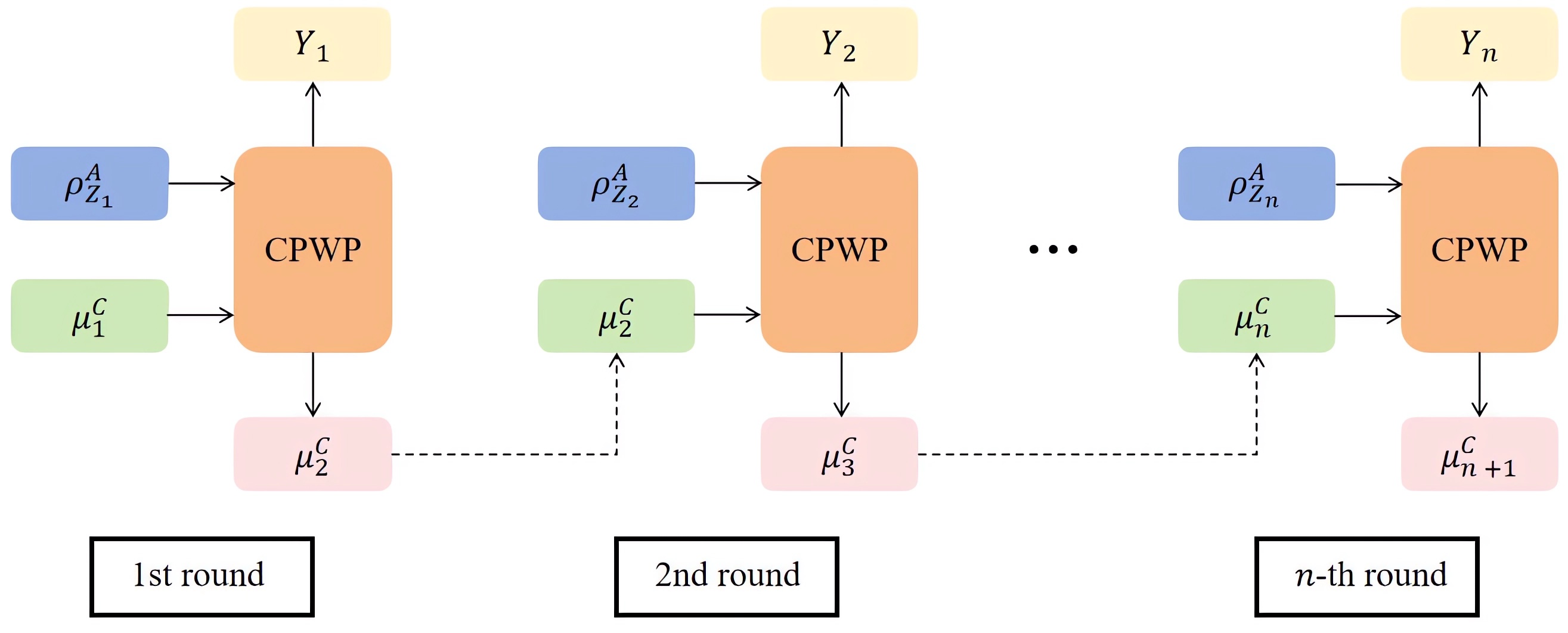}
	\caption{\textbf{General state discrimination framework under CPWP protocol assisted by quantum memories and catalysts.} In each round, the input state $\rho_{Z_j}^A$ and the memory state $\mu_j^C$ are jointly processed by a CPWP instrument, producing a classical guess $Y_j$ and an updated memory state $\mu_{j+1}^C$. The updated memory is reused in the next round, thereby allowing information from previous rounds to affect later discrimination strategies. In the catalytic setting, the memory state remains unchanged in every round, i.e., $\mu_1^C=\mu_j^C$ for all $j \in \left\{2,\cdots,n\right\}$.}
	\label{fig2}
\end{figure*}

After defining these variables, we can introduce the definition of achievable success rate for discrimination.
\begin{definition}
	A rate $r\in[0,1]$ is said to be achievable if for every $\varepsilon>0$ and every positive integer $m$, there exists an integer $n\geqslant m$ such that
	\begin{equation}
		\Pr(S_n\geqslant rn)\geqslant 1-\varepsilon.
	\end{equation}
	Moreover, the optimal success rate $R$ is obtained by taking the supremum over all achievable rates $r$.
\end{definition}
Hereinafter, we will use $R_m$ to denote the optimal memory-assisted PWF discrimination rate and $R_c$ to denote the optimal catalyst-assisted PWF discrimination rate.
Before introducing the main theorem in this section, we first give the following lemma.

\begin{lemma}\label{lemma1}
	Let $\Delta_A$ be a Hermitian operator, $\gamma_C$ be a PWF state. Then
	\begin{equation}
		\left\|\Delta_A\otimes\gamma_C\right\|_{\mathbf{PWF}}
		=
		\left\|\Delta_A\right\|_{\mathbf{PWF}}.
	\end{equation}
\end{lemma}

\begin{proof}
	We first prove $\left\|\Delta_A\otimes\gamma_C\right\|_{\mathbf{PWF}} \geqslant \left\|\Delta_A\right\|_{\mathbf{PWF}}$. 
	Let $\mathbf{E}^A=\{E_i^A\}_i$ be an arbitrary PWF POVM, we set $\mathbf{E}^{AC}=\{E_i^A\otimes \mathds{1}_C\}_i$. 
	It is easy to see that $\mathbf{E}^{AC}$ is also a PWF POVM. 
	Since
	\begin{equation}
		\begin{aligned}
			\left\|\Delta_A\otimes\gamma_C\right\|_{\mathbf{PWF}} &\geqslant \sum_i\left| \mathrm{Tr}\left[(E_i^A\otimes \mathds{1}_C) (\Delta_A\otimes\gamma_C)\right] \right| \\ &= \sum_i \left| \mathrm{Tr}(E_i^A\Delta_A) \right|.
		\end{aligned}
	\end{equation}
	Then we take the maximum over $\{E_i^A\}_i$, we can get
	\begin{equation}
		\left\|\Delta_A\otimes\gamma_C\right\|_{\mathbf{PWF}} \geqslant \left\|\Delta_A\right\|_{\mathbf{PWF}}.
	\end{equation}
	
	Subsequently, we prove $\left\|\Delta_A\otimes\gamma_C\right\|_{\mathbf{PWF}} \leqslant \left\|\Delta_A\right\|_{\mathbf{PWF}}$. 
	Let $\mathcal{E}^{AC\rightarrow B}$ be a measurement channel associated to $\mathbf{E}^{AC}=\{E_i^{AC}\}_i$, then we have
	\begin{equation}
		\mathcal{E}^{AC\rightarrow B}(\rho_{AC}) = \sum_i \mathrm{Tr}\left(E_i^{AC}\rho_{AC}\right)\ket{i}\bra{i},
	\end{equation}
	where $\{\ket{i}\}$ is an orthonormal basis on system $B$.
	Since $\mathbf{E}^{AC}$ is a PWF POVM, the measurement channel $\mathcal{E}^{AC\rightarrow B}$ is a CPWP channel.
	Then we consider a special preparation channel $\mathcal{P}_{\gamma_C}^{A\rightarrow AC}$ of the PWF state $\gamma_C$, which is defined as follows
	\begin{equation}
		\mathcal{P}_{\gamma_C}^{A\rightarrow AC}(\rho_A) = \rho_A \otimes\gamma_C.
	\end{equation}
	Because $\gamma_C$ is a PWF state, this special preparation channel $\mathcal{P}_{\gamma_C}^{A\rightarrow AC}$ is a CPWP channel for any PWF input state $\rho_A$. 
	Therefore, the composed channel $\mathcal{N}^{A\rightarrow B} \coloneqq \mathcal{E}^{AC\rightarrow B}\circ \mathcal{P}_{\gamma_C}^{A\rightarrow AC}$ is a CPWP channel. 
	Moreover, $\mathcal{N}^{A\rightarrow B}$ is a measurement channel, so there exists a PWF POVM $\mathbf{F}^A=\{F_i^A\}_i$ such that
	\begin{equation}
		\mathcal{F}^{A\rightarrow B}(\rho_A) = \sum_i \mathrm{Tr}(F_i^A \rho_A)\ket{i}\bra{i}.
	\end{equation}
	According to the definition of the composed channel $\mathcal{F}^{A\rightarrow B}$, we have
	\begin{equation}
		\sum_i \left|\mathrm{Tr}\left[E_i^{AC}(\Delta_A\otimes\gamma_C) \right] \right| = \sum_i \left| \mathrm{Tr}(F_i^A\Delta_A) \right| \leqslant \left\|\Delta_A\right\|_{\mathbf{PWF}}.
	\end{equation}
	Then we take the maximum over $\{E_i^{AC}\}_i$, we can get
	\begin{equation}
		\left\|\Delta_A\otimes\gamma_C\right\|_{\mathbf{PWF}} \leqslant \left\|\Delta_A\right\|_{\mathbf{PWF}}.
	\end{equation}
		This completes the proof.
\end{proof}

Next, we present the main theorem of this section. 
We prove that neither quantum memories nor quantum catalysts can improve the discrimination power of CSMs when distinguishing the pair of PWF states.

\begin{theorem}\label{theorem 3}
	Let $\rho_0,\rho_1$ be two PWF states with prior probabilities $p_0,p_1>0$ and $p_0+p_1=1$. 
	Then we have
	\begin{equation}
		R_m\left(\{p_i,\rho_i\}_{i=0}^{1}\right) = p_{\mathrm{succ}}^{\mathbf{PWF}}\left(\{p_i,\rho_i\}_{i=0}^{1}\right),
	\end{equation}
	i.e., quantum memories cannot improve the discrimination power of CSMs for distinguishing the pair of PWF states.
\end{theorem}

\begin{proof}
	First, we prove $R_m\left(\{p_i,\rho_i\}_{i=0}^{1}\right) \geqslant p_{\mathrm{succ}}^{\mathbf{PWF}}\left(\{p_i,\rho_i\}_{i=0}^{1}\right)$.
	This inequality is obvious, since the memories can always be ignored and one can directly perform an optimal PWF POVM. 
	Therefore, the optimal memory-assisted rate cannot be lower than the success probability achieved by PWF measurements.
	
	Subsequently, we prove $R_m\left(\{p_i,\rho_i\}_{i=0}^{1}\right) \leqslant p_{\mathrm{succ}}^{\mathbf{PWF}}\left(\{p_i,\rho_i\}_{i=0}^{1}\right)$.
	We will prove this inequality by contradiction.
	We optimize over protocols with a finite-dimensional reusable memory. 
	The memory dimension is not allowed to scale with $n$.
	Let $p_* = p_{\mathrm{succ}}^{\mathbf{PWF}} \left(\{p_i,\rho_i\}_{i=0}^{1}\right)$ and $r = R_m\left(\{p_i,\rho_i\}_{i=0}^{1}\right)$.
	Suppose that there exists a $d$-dimensional reusable quantum memory state $\mu_C$ and an adaptive sequence of CPWP instruments such that
	\begin{equation}
		r>p_*.
	\end{equation}
	Then we choose an achievable rate $r$ and $\delta>0$ such that
	\begin{equation}\label{r>p*}
		r>p_*+\delta.
	\end{equation}
	According to the definition of the achievable rate, for every $\varepsilon>0$ and every positive integer $m$, there exists an integer $n\geqslant m$ such that
	\begin{equation}
		\Pr\left(S_n\geqslant rn\,\middle|\,\mu_C\right) \geqslant 1-\varepsilon.
	\end{equation}
	Combining this inequality with Eq. \eqref{r>p*}, we can get
	\begin{equation}\label{memory probability}
		\Pr\left(\frac{S_n}{n}\geqslant p_*+\delta\,\middle|\, \mu_C \right) \geqslant 1-\varepsilon.
	\end{equation}
	
	Next, let $\gamma_C$ be the maximally mixed state, i.e., $\gamma_C \coloneqq \mathds{1}/d$. 
	It is easy to see that $\gamma_C$ is a stabilizer state, so it is also a PWF state.
	Since $\gamma_C$ is full rank, $\mu_C$ and $\gamma_C$ are not orthogonal.
	Then, we consider a binary QSD task with equal prior probabilities, i.e., using the global measurements to distinguish the pair of states $\{\mu_C, \gamma_C\}$ with prior probabilities $1/2$. 
	Therefore, the optimal success probability for distinguishing $\mu_C$ from $\gamma_C$ satisfies 
	\begin{equation}
		p_{\mathrm{opt}}(\mu_C,\gamma_C)= \frac{1}{2} + \frac{1}{4} \left\|\mu_C-\gamma_C\right\|_1 <1.
	\end{equation}
	
	Assume that the initial state of the memory state is either $\mu_C$ or $\gamma_C$ with equal probability $1/2$. 
	In each round $j$, the variable $Z_j\in\{0,1\}$ is sampled according to the prior probabilities $p_0,p_1$. 
	If $Z_j=i$, then the PWF state $\rho_i$ is prepared. 
	Then the $j$-th CPWP instrument is applied to the joint system $AC$, producing a classical guess $Y_j\in\{0,1\}$ and an updated memory state. 
	Next, we define
	\begin{equation}
		X_j=
		\begin{cases}
			1, & Y_j=Z_j,\\
			0, & Y_j\neq Z_j.
		\end{cases}
	\end{equation}
	After $n$ rounds, we compute $S_n=\sum_{j=1}^{n}X_j$. 
	If the following inequality is satisfied
	\begin{equation}\label{decision rule}
		\frac{S_n}{n}\geqslant p_*+\delta,
	\end{equation}
	we guess that the initial memory state is $\mu_C$.
	Otherwise, we guess that the initial memory state is $\gamma_C$.
	Subsequently, we will consider the following two cases.
	
	(i)~If the initial memory state is $\mu_C$. 
	Then combined with Eq. \eqref{memory probability} and  Eq. \eqref{decision rule} , we can get the probability of guessing correctly is at least $1-\varepsilon$, i.e.,
	\begin{equation}\label{probability mu}
		P_{\mathrm{corr}}(\mu_C)\geqslant 1-\varepsilon.
	\end{equation}
	
	(ii)~If the initial memory state is $\gamma_C$. 
	Let $\mathcal{F}_j$ denote the history of classical guess generated by CPWP instruments up to the end of round $j$, $\gamma_C(\mathcal{F}_j)$ denote the conditional memory state which depends on the history $\mathcal{F}_j$ at the beginning of the $(j+1)$-th round.
	Since $\rho_0, \rho_1, \gamma_C$ are PWF states, and each instrument is CPWP instrument, the conditional memory state $\gamma_C(\mathcal{F}_j)$ is a PWF state. 
	At this point, the pair of input states to be distinguished is $\{\rho_0\otimes\gamma_C(\mathcal{F}_j), \rho_1\otimes\gamma_C(\mathcal{F}_j)\}$ with prior probabilities $p_0,p_1$.
	Then we can get that
	\begin{equation}
		\left\|\left(p_0\rho_0-p_1\rho_1\right)\otimes \gamma_C(\mathcal{F}_j) \right\|_{\mathbf{PWF}} = \left\|p_0\rho_0-p_1\rho_1\right\|_{\mathbf{PWF}},
	\end{equation}
	where the equality follows from the conditional memory state $\gamma_C(\mathcal{F}_j)$ is a PWF state and Lemma \ref{lemma1}.
	Therefore, we have $p_{\mathrm{succ}}^{\mathbf{PWF}}\left(\{p_i,\rho_i \otimes \gamma_C(\mathcal{F}_j) \}_{i=0}^{1}\right) = p_{\mathrm{succ}}^{\mathbf{PWF}}\left(\{p_i,\rho_i\}_{i=0}^{1}\right)$.
	So the success probability of the guess in the $(j+1)$-th round cannot exceed $p_*$, i.e.,
	\begin{equation}\label{conditional success bound}
		\Pr(X_{j+1}=1\mid \mathcal{F}_j) \leqslant p_*.
	\end{equation}
	
	Next, we define the random variable $C_j\coloneqq S_j-jp_*$ for $1\leqslant j \leqslant n$ and $C_0=0$.
	Using Eq.~\eqref{conditional success bound}, we obtain
	\begin{equation}
		\begin{aligned}
			\mathbb{E}\left(C_{j+1}\mid \mathcal{F}_j\right) &= \mathbb{E}\left(S_j+X_{j+1}-(j+1)p_*\,\middle|\,\mathcal{F}_j\right) \\ &= C_j+\mathbb{E}\left(X_{j+1}\mid \mathcal{F}_j\right)-p_*\\ &= C_j+\Pr(X_{j+1}=1\mid \mathcal{F}_j)-p_* \\ &\leqslant C_j,
		\end{aligned}
	\end{equation}
	where the first equality follows from the definition of $C_{j+1}$, the second equality follows from $C_{j+1} = (S_j-jp_*)+X_{j+1}-p_* = C_j+X_{j+1}-p_*$ and $\mathbb{E}(C_j\mid \mathcal{F}_j) = C_j$, the last equality follows from $\mathbb{E}\left(X_{j+1}\mid \mathcal{F}_j\right) = 0 \cdot \Pr(X_{j+1}=0\mid \mathcal{F}_j) + 1 \cdot \Pr(X_{j+1}=1\mid \mathcal{F}_j) = \Pr(X_{j+1}=1\mid \mathcal{F}_j)$, the inequality follows from Eq. \eqref{conditional success bound}.
	Hence, $\{C_j\}_{j=0}^{n}$ is a supermartingale \cite{R1995}. 
	Since $C_{j+1}-C_j=X_{j+1}-p_*$, $X_{j+1}\in\{0,1\}$ and $0\leqslant p_*\leqslant 1$, we have
	\begin{equation}
		\left|C_{j+1}-C_j\right|\leqslant 1.
	\end{equation}
	According to Azuma's inequality for supermartingales \cite{A1967}, we have
	\begin{equation}
		\Pr\left(C_n\geqslant n\delta\,\middle|\,\gamma_C\right) = 	\Pr\left(\frac{S_n}{n}\geqslant p_*+\delta\,\middle|\,\gamma_C \right) \leqslant \exp\left(-\frac{n\delta^2}{2}\right)
	\end{equation}
	for all $n$ and $\delta > 0$.
	Furthermore, if Eq. \eqref{decision rule} holds, we guess the initial memory state is $\mu_C$, and this is a wrong guess.
	Thus, when the initial memory state is $\gamma_C$, the probability of guessing correctly is
	\begin{equation}\label{probability gamma}
		P_{\mathrm{corr}}(\gamma_C) = 1-\Pr\left(C_n\geqslant n\delta\,\middle|\,\gamma_C\right) \geqslant 1-\exp\left(-\frac{n\delta^2}{2}\right).
	\end{equation}
	
	Combining Eqs.~\eqref{probability mu} and \eqref{probability gamma}, the total success probability for distinguishing $\mu_C$ from $\gamma_C$ satisfies
	\begin{equation}
		\begin{aligned}
			P_{\mathrm{corr}}(\mu_C,\gamma_C) &= \frac{1}{2}P_{\mathrm{corr}}(\mu_C) + \frac{1}{2}P_{\mathrm{corr}}(\gamma_C) \\ &\geqslant \frac{1}{2}(1-\varepsilon) + \frac{1}{2} \left[ 1-\exp\left(-\frac{n\delta^2}{2}\right)\right].
		\end{aligned}
	\end{equation}
	Since $\varepsilon>0$ can be chosen arbitrarily small and $n$ can be chosen arbitrarily large, we can get $P_{\mathrm{corr}}(\mu_C,\gamma_C) \rightarrow 1$ when $\varepsilon \rightarrow 0$ and $n \rightarrow + \infty$. 
	This contradicts the fact that two nonorthogonal quantum states cannot be distinguished with unit success probability.
	Therefore, we can get that $R_m\left(\{p_i,\rho_i\}_{i=0}^{1}\right) \leqslant p_{\mathrm{succ}}^{\mathbf{PWF}}\left(\{p_i,\rho_i\}_{i=0}^{1}\right)$.
	
	Combining $R_m\left(\{p_i,\rho_i\}_{i=0}^{1}\right) \geqslant p_{\mathrm{succ}}^{\mathbf{PWF}}\left(\{p_i,\rho_i\}_{i=0}^{1}\right)$ with $R_m\left(\{p_i,\rho_i\}_{i=0}^{1}\right) \leqslant p_{\mathrm{succ}}^{\mathbf{PWF}}\left(\{p_i,\rho_i\}_{i=0}^{1}\right)$, we have
	\begin{equation}
		R_m\left(\{p_i,\rho_i\}_{i=0}^{1}\right) = p_{\mathrm{succ}}^{\mathbf{PWF}}\left(\{p_i,\rho_i\}_{i=0}^{1}\right).
	\end{equation}
	This completes the proof.
\end{proof}

So far, we have proved that quantum memories cannot improve the discrimination power of CSMs when distinguishing the pair of PWF states.
Since catalyst-assisted discrimination is a special case of memory-assisted discrimination, in which the memory state remains unchanged after each round, we immediately obtain the following corollary.

\begin{corollary}\label{corollary 2}
	Let $\rho_0,\rho_1$ be two PWF states with prior probabilities $p_0,p_1>0$ and $p_0+p_1=1$. 
	Then we have
	\begin{equation}
		R_c\left(\{p_i,\rho_i\}_{i=0}^{1}\right) = p_{\mathrm{succ}}^{\mathbf{PWF}}\left(\{p_i,\rho_i\}_{i=0}^{1}\right),
	\end{equation}
	i.e., quantum catalysts cannot improve the discrimination power of CSMs for distinguishing the pair of PWF states.
\end{corollary}

\begin{proof}
	First, we prove $R_c\left(\{p_i,\rho_i\}_{i=0}^{1}\right) \geqslant p_{\mathrm{succ}}^{\mathbf{PWF}}\left(\{p_i,\rho_i\}_{i=0}^{1}\right)$.
	This inequality is obvious, since the catalysts can always be ignored and one can directly perform an optimal PWF POVM. 
	Therefore, the optimal catalyst-assisted rate cannot be lower than the success probability achieved by PWF measurements.
	
	Subsequently, we prove $R_c\left(\{p_i,\rho_i\}_{i=0}^{1}\right) \leqslant p_{\mathrm{succ}}^{\mathbf{PWF}}\left(\{p_i,\rho_i\}_{i=0}^{1}\right)$.
	Since the catalyst state can be regarded as a memory state that is returned exactly to its initial state and remains uncorrelated with the output after each round, we have
	\begin{equation}
		R_c\left(\{p_i,\rho_i\}_{i=0}^{1}\right) \leqslant R_m\left(\{p_i,\rho_i\}_{i=0}^{1}\right).
	\end{equation}
	According to Theorem \ref{theorem 3}, we have
	\begin{equation}
		R_m\left(\{p_i,\rho_i\}_{i=0}^{1}\right) = p_{\mathrm{succ}}^{\mathbf{PWF}}\left(\{p_i,\rho_i\}_{i=0}^{1}\right).
	\end{equation}
	Therefore, we can get $R_c\left(\{p_i,\rho_i\}_{i=0}^{1}\right) \leqslant p_{\mathrm{succ}}^{\mathbf{PWF}}\left(\{p_i,\rho_i\}_{i=0}^{1}\right)$.

	Combining $R_c\left(\{p_i,\rho_i\}_{i=0}^{1}\right) \geqslant p_{\mathrm{succ}}^{\mathbf{PWF}}\left(\{p_i,\rho_i\}_{i=0}^{1}\right)$ with $R_c\left(\{p_i,\rho_i\}_{i=0}^{1}\right) \leqslant p_{\mathrm{succ}}^{\mathbf{PWF}}\left(\{p_i,\rho_i\}_{i=0}^{1}\right)$, we have
	\begin{equation}
		R_c\left(\{p_i,\rho_i\}_{i=0}^{1}\right) = p_{\mathrm{succ}}^{\mathbf{PWF}}\left(\{p_i,\rho_i\}_{i=0}^{1}\right).
	\end{equation}
	This completes the proof.
\end{proof}

So far, we have proved that neither quantum memories nor quantum catalysts can improve the discrimination power of CSMs when discriminating a pair of PWF states. 
Meanwhile, Refs. \cite{SLN+2022,PS2025,YDY2012} have shown that quantum catalysts and quantum memories can enhance the success probability of local state discrimination when discriminating a pair of entangled states in the entanglement resource theory. 
However, we are unable to find a state that is arbitrarily close to a PWF state and can still provide a fixed amount of magic resources in each round to enhance the success probability of discrimination in the magic resource theory. 
Therefore, we cannot prove whether quantum catalysts and quantum memories can enhance the success probability of CSM discrimination for a pair of non-PWF states. 
Nevertheless, we conjecture that this is true, and leave this question for future research. 

To further reveal the different roles of consumable and reusable quantum resources in improving the discrimination power of CSMs, we construct an ensemble consisting of two PWF states and compare the optimal discrimination performance assisted by different types of quantum resources.

\begin{corollary}\label{corollary consumable and reusable resources}
	Let
	\begin{equation}
		\widetilde{\rho}_0=\frac{\mathds{1}+\ket{S}\bra{S}}{4},\quad\widetilde{\rho}_1=\frac{\mathds{1}-\ket{S}\bra{S}}{2}.
	\end{equation}
	Consider the ensemble $\widetilde{\Pi}=\{(2/5,\widetilde{\rho}_0), (3/5,\widetilde{\rho}_1)\}$. 
	Then, for any nonnegative integer $n$, we have
	\begin{equation}
		p_{\mathrm{succ}}^{\mathbf{PWF}+(\ket{S}\bra{S})^{\otimes n}}(\widetilde{\Pi})=\min\left\{\frac{4}{5},\frac{1}{2}+\frac{1}{10}\left(\frac{5}{3}\right)^n\right\}.
	\end{equation}
	However, for finite-dimensional quantum catalysts and fixed finite-dimensional quantum memories, we have
	\begin{equation}
		R_c(\widetilde{\Pi})=R_m(\widetilde{\Pi})=p_{\mathrm{succ}}^{\mathbf{PWF}}(\widetilde{\Pi})=\frac{3}{5}.
	\end{equation}
	Thus, consumable magic resources enable CSMs to achieve the optimal success probability of global measurements, whereas neither quantum catalysts nor quantum memories improve the optimal asymptotic success rate.
\end{corollary}

\begin{proof}
	First, we prove that both $\widetilde{\rho}_0$ and $\widetilde{\rho}_1$ are PWF states.
	According to Eq. \eqref{Wigner function of Strange state}, by direct calculation, we have
	\begin{equation}
		\begin{aligned}
			W_{\widetilde{\rho}_0}(\mathbf{u}_-)&=0,\quad W_{\widetilde{\rho}_0}(\mathbf{u}_+)=\frac{1}{8},\\ W_{\widetilde{\rho}_1}(\mathbf{u}_-)&=\frac{1}{3},\quad W_{\widetilde{\rho}_1}(\mathbf{u}_+)=\frac{1}{12}.
		\end{aligned}
	\end{equation}
	Thus, both $\widetilde{\rho}_0$ and $\widetilde{\rho}_1$ are PWF states.
	
	Subsequently, we consider the ensemble $\Pi=\{(1/3,\rho_0),(2/3,\rho_1)\}$ given in Proposition \ref{proposition 2}. 
	By direct calculation, we have
	\begin{equation}
		\frac{2}{5}\widetilde{\rho}_0-\frac{3}{5}\widetilde{\rho}_1=\frac{2}{5}\ket{S}\bra{S}-\frac{1}{5}\mathds{1}=\frac{3}{5}\left(\frac{1}{3}\rho_0-\frac{2}{3}\rho_1\right)=-\frac{1}{5}A_{\mathbf{u}_-}.
	\end{equation}
	According to the definition of the optimal magic-assisted average success probability and the homogeneity of the distinguishability norm, we have
	\begin{equation}
		\begin{aligned}
			p_{\mathrm{succ}}^{\mathbf{PWF}+(\ket{S}\bra{S})^{\otimes n}}(\widetilde{\Pi})&=\frac{1}{2}\left[1+\left\|\left(\frac{2}{5}\widetilde{\rho}_0-\frac{3}{5}\widetilde{\rho}_1\right)\otimes(\ket{S}\bra{S})^{\otimes n}\right\|_{\mathbf{PWF}}\right]\\&=\frac{1}{2}+\frac{3}{10}\left\|\left(\frac{1}{3}\rho_0-\frac{2}{3}\rho_1\right)\otimes(\ket{S}\bra{S})^{\otimes n}\right\|_{\mathbf{PWF}}\\&=\frac{1}{5}+\frac{3}{5}p_{\mathrm{succ}}^{\mathbf{PWF}+(\ket{S}\bra{S})^{\otimes n}}(\Pi).
		\end{aligned}
	\end{equation}
	According to Corollary \ref{corollary Strange assisted discrimination}, we can get
	\begin{equation}
		\begin{aligned}
			p_{\mathrm{succ}}^{\mathbf{PWF}+(\ket{S}\bra{S})^{\otimes n}}(\widetilde{\Pi})&=\frac{1}{5}+\frac{3}{5}\min\left\{1,\frac{1}{2}+\frac{1}{6}\left(\frac{5}{3}\right)^n\right\}\\&=\min\left\{\frac{4}{5},\frac{1}{2}+\frac{1}{10}\left(\frac{5}{3}\right)^n\right\}.
		\end{aligned}
	\end{equation}
	In particular, taking $n=0$, we have
	\begin{equation}
		p_{\mathrm{succ}}^{\mathbf{PWF}}(\widetilde{\Pi})=\frac{3}{5}.
	\end{equation}
	
	Moreover, according to the Holevo-Helstrom theorem, we have
	\begin{equation}
		\begin{aligned}
			p_{\mathrm{succ}}^{\mathbf{GLOBAL}}(\widetilde{\Pi})&=\frac{1}{2}\left(1+\left\|\frac{2}{5}\widetilde{\rho}_0-\frac{3}{5}\widetilde{\rho}_1\right\|_1\right)\\&=\frac{1}{2}\left(1+\frac{1}{5}\|A_{\mathbf{u}_-}\|_1\right)\\&=\frac{4}{5},
		\end{aligned}
	\end{equation}
	where the last equality follows from the fact that the eigenvalues of the qutrit phase space point operator $A_{\mathbf{u}_-}$ are $-1,1,1$.
	
	Finally, since both $\widetilde{\rho}_0$ and $\widetilde{\rho}_1$ are PWF states, according to Theorem \ref{theorem 3} and Corollary \ref{corollary 2}, we have
	\begin{equation}
		R_c(\widetilde{\Pi})=R_m(\widetilde{\Pi})=p_{\mathrm{succ}}^{\mathbf{PWF}}(\widetilde{\Pi})=\frac{3}{5}.
	\end{equation}
	This completes the proof.
\end{proof}

\section{Conclusions and discussions}
In this paper, we study how to improve the discrimination power of CSMs within the framework of the magic resource theory.
Specifically, we consider three possible approaches to improving their discrimination power, including adding magic resources to restricted measurements, using quantum catalysts, and using quantum memories.
First, we consider QSD under restricted measurements assisted by magic resources. 
We establish an equivalence between PWF POVMs and CPWP measurement channels, thereby transforming the problem of implementing measurements beyond CSMs into a channel simulation problem using CPWP channels and magic states. 
Next, we derive lower and upper bounds on the exact channel simulation cost and provide a concrete example for which these bounds coincide, thereby exactly determining the amount of magic resources required to simulate the magic measurement channel. 
Moreover, we prove that consumable magic resources can improve the discrimination power of CSMs and show that three copies of the Strange state are necessary and sufficient for perfect discrimination of a specific ensemble using CSMs. 
Furthermore, we analytically determine the optimal success probability of discriminating this ensemble assisted by an arbitrary number of copies of the Strange state, thereby exactly characterizing the relationship between the amount of magic resources and the optimal discrimination success probability. 
For more general magic-assisted binary QSD tasks with equal prior probabilities, we formulate the corresponding SDP and derive its dual problem. 
Second, we investigate whether quantum catalysts and quantum memories can improve the discrimination power of CSMs. 
We establish a no-go theorem showing that, for a pair of PWF states, neither finite-dimensional quantum catalysts nor fixed finite-dimensional quantum memories can improve the optimal asymptotic success rate under CSMs.
Finally, by constructing a specific ensemble consisting of two PWF states, we quantitatively compare the roles of consumable and reusable quantum resources in enhancing the discrimination power of CSMs. 
For this ensemble, the optimal success probability under CSMs increases from $3/5$ without magic assistance to $2/3$ and $7/9$ with one and two copies of the Strange state, respectively, and reaches the optimal success probability of global measurements, i.e., $4/5$, with three copies. 
In contrast, the optimal asymptotic success rate assisted by finite-dimensional quantum catalysts or fixed finite-dimensional quantum memories remains $3/5$. 
This result demonstrates that, within the discrimination framework considered in this paper, consumable magic resources can bridge the gap in discrimination power between CSMs and global measurements, whereas reusable quantum resources cannot provide a sustained discrimination advantage. 
These results reveal the different operational roles of consumable and reusable quantum resources in restricted quantum state discrimination and provide further insight into the relationship between magic resources and the discrimination power of quantum measurements.

Several problems remain worthy of further investigation. 
First, although we have shown that the channel simulation bounds can be tight for a specific non-CPWP measurement channel, it remains open to determine tighter and more computable bounds for general target channels.
Recently, the stabilizer R\'enyi entropy has attracted considerable attention due to its advantage of easy-to-compute \cite{LOH2022,WL2023}. 
Can it be used to derive tighter upper and lower bounds on the simulation cost?
Second, it would be interesting to further establish an SDP hierarchy for the optimal magic cost of magic-assisted discrimination, similar to the method used for entanglement-assisted discrimination in Ref. \cite{ZZL+2025}. 
Subsequently, the no-go theorem in this paper applies to the discrimination of a pair of PWF states using finite-dimensional memories or catalysts. 
Whether catalysts or memories can improve the discrimination power for non-PWF states, or provide an advantage in more general state discrimination tasks, remains an interesting problem. 
Finally, investigating whether magic activators \cite{C2011} can improve the discrimination power of CSMs is also a worthwhile direction for future research.

\section*{ACKNOWLEDGMENTS}
This paper was supported by National Natural Science Foundation of China (Grant No.12471437), Postgraduate Education Comprehensive Reform Research and Practice Project of Shaanxi Province in 2025 (Grant No. YJSZG2025062) and General Projects of Degree and Graduate Education Research in Shaanxi Province in 2026 (Grant No. SXGERC2026058).

\appendix
\section*{APPENDIX A}
According to the primal problem
\begin{equation}
	\begin{aligned}
		&\max_{E} \quad &&\frac{1}{2}+\frac{1}{2}\mathrm{Tr}\{E[(\rho_0 - \rho_1)\otimes \omega^{\otimes n}]\}, \\ &\mathrm{subject~to} \quad && 0\leqslant E\leqslant\mathds{1},  \\ &&& 0\leqslant \mathrm{Tr}(EA_{\mathbf{u}})\leqslant 1 \quad \forall u,
	\end{aligned}
\end{equation}
where $E$ is a PWF POVM, $\omega$ is a magic state, the constraints follow from the properties of PWF POVM.

Next, we introduce the dual variables $Y\geqslant 0$ for the constraint \(\mathds{1}-E\geqslant 0\), and nonnegative real numbers $\alpha_{\mathbf u}\geqslant 0, \beta_{\mathbf u}\geqslant 0$ for the constraints$ \mathrm{Tr}(EA_{\mathbf u})\geqslant 0, 1-\mathrm{Tr}(EA_{\mathbf u})\geqslant 0$,respectively. 
Then the Lagrange function is
\begin{equation}
	\begin{aligned}
		L(E;Y,\alpha_{\mathbf u},\beta_{\mathbf u}) =&\frac12+\mathrm{Tr}(\frac{1}{2}\left[(\rho_0-\rho_1)\otimes\omega^{\otimes n}\right]E) +\mathrm{Tr}\!\left[Y(\mathds{1}-E)\right] \\ &+\sum_{\mathbf u}\alpha_{\mathbf u} \mathrm{Tr}(EA_{\mathbf u}) +\sum_{\mathbf u}\beta_{\mathbf u} \left[ 1-\mathrm{Tr}(EA_{\mathbf u}) \right].
	\end{aligned}
\end{equation}

We can write the above Lagrange function as
\begin{equation}
	\begin{aligned}
		L(&E;Y,\alpha_{\mathbf u},\beta_{\mathbf u}) = \frac12+\mathrm{Tr}(Y) +\sum_{\mathbf u}\beta_{\mathbf u} \\ &+ \mathrm{Tr}\left[ E\left(\frac{1}{2}\left[(\rho_0-\rho_1)\otimes\omega^{\otimes n}\right]-Y+\sum_{\mathbf u}(\alpha_{\mathbf u}-\beta_{\mathbf u})A_{\mathbf u}\right) \right].
	\end{aligned}
\end{equation}

Since the maximization is over $E\geqslant 0$, the supremum of the
Lagrange function is finite if and only if
\begin{equation}
	Y+\sum_{\mathbf u}
	(\beta_{\mathbf u}-\alpha_{\mathbf u})A_{\mathbf u}
	\geqslant \frac{1}{2}\left[(\rho_0-\rho_1)\otimes\omega^{\otimes n}\right].
\end{equation}

Therefore, the dual problem is
\begin{equation}
	\begin{aligned}
		\min_{Y,\alpha_{\mathbf u},\beta_{\mathbf u}} \quad & \frac12+\mathrm{Tr}(Y) +\sum_{\mathbf u}\beta_{\mathbf u} \\ \mathrm{subject\ to}\quad & Y+\sum_{\mathbf u} (\beta_{\mathbf u}-\alpha_{\mathbf u})A_{\mathbf u} \geqslant \frac12\big((\rho_0-\rho_1)\otimes\omega^{\otimes n}\big),\\ & Y\geqslant 0,\\ & \alpha_{\mathbf u}\geqslant 0,\quad \beta_{\mathbf u}\geqslant 0, \qquad \forall \mathbf u .
	\end{aligned}
\end{equation}
Since the primal problem satisfies Slater's condition, for example by choosing $E = \mathds{1}/2$, strong duality holds.

\end{document}